\documentclass[12pt, draftclsnofoot, onecolumn]{IEEEtran} 

\usepackage{amsfonts}
\usepackage{mathrsfs}
\usepackage{amssymb,amsmath}
\usepackage[compress]{cite}
\usepackage{bm}
\usepackage{algorithm}
\usepackage{algorithmic}
\usepackage{epsfig,subfigure}
\usepackage{color}
\usepackage{multirow}
\usepackage{amsthm}
\usepackage{array}
\usepackage{tabularx}
\usepackage{multicol}
\usepackage{booktabs}
\usepackage{cite}
\usepackage{stfloats}
\usepackage{balance}
\usepackage{graphics}

\newtheorem{lemma}{Lemma}

\newtheorem{proposition}{Theorem}

\begin{document}
	\title{ Multi-Antenna Aided Secrecy Beamforming Optimization for Wirelessly Powered HetNets}
  \author{Shiqi Gong, Shaodan Ma, Chengwen Xing, Yonghui Li,~\IEEEmembership{Fellow,~IEEE}, and Lajos Hanzo,~\IEEEmembership{Fellow,~IEEE} }
	\maketitle
	\vspace{-13mm}
	\begin{abstract}
		The new paradigm  of  wirelessly powered  two-tier heterogeneous networks (HetNets) is considered in this paper.  Specifically, the femtocell base station (FBS) is powered by a power beacon (PB)  and transmits confidential information to  a legitimate femtocell user (FU) in the presence of a potential eavesdropper (EVE) and  a macro base station (MBS).  In this scenario,  we investigate   the secrecy beamforming design  under  three different levels of FBS-EVE channel state information (CSI), namely,
		 the perfect,  imperfect  and  completely unknown FBS-EVE CSI.  Firstly, given the  perfect global CSI at the FBS, the PB energy covariance matrix,  the FBS information covariance matrix and the  time splitting  factor are jointly optimized  aiming for perfect  secrecy rate maximization.   Upon assuming   the   imperfect  FBS-EVE  CSI,    the worst-case  and  outage-constrained SRM problems corresponding to  deterministic and  statistical  CSI errors are investigated,  respectively.   Furthermore, considering the more  realistic case of   unknown FBS-EVE CSI ,   the  artificial noise (AN) aided secrecy beamforming design is  studied. Our analysis reveals   that  for all above cases  both  the   optimal PB energy  and  FBS  information secrecy beamformings are of rank-1. Moreover,   for all considered cases of FBS-EVE CSI, the closed-form   PB energy   beamforming  solutions  are   available   when  the    cross-tier interference constraint is inactive. Numerical simulation results  demonstrate the secrecy performance advantages of all  proposed secrecy beamforming designs compared to the adopted  baseline algorithms.
	\end{abstract}
%
\vspace{-1.5mm}

	\section{Introduction}
With the proliferation of smart devices  and  data-hungry applications,  establishing  ubiquitous, high-throughput and  secure  communications   is   gaining increased  importance in next-generation  systems \cite{Jiang:2016hdba,Boccardi:2014kzba,Hossain:2014ddba}. The traditional macrocells  generally have poor performance in terms of  indoor coverage and  cell edge rate. To tackle this issue,  heterogenous networks (HetNets) have emerged as a promising  next-generation architecture, which  are  generally  supported  by  heterogenous  base  stations (BSs)  having  different  service coverages \cite{Ghosh:2012wo,Jo:2012dc}. Specifically, the macrocell  base station (MBS)  can   provide open  access  and  wide coverage up to dozens of  kilometers, while
	the low-power femtocell base station (FBS) and picocell base station (PBS)  are typically deployed  in  indoor environments and near to femtocell users (FUs) and picocell users (PUs), respectively. As pointed out in \cite{Shariat:2016tvba},  the ultra-dense  deployment  of
	femtocells  is recognized as an effecient technique  to realize 1000 times increase in wireless data rate for.  However, due to  the high spatial spectrum reuse in HetNets and the dense deployment of FBSs and PBSs,  cross-tier  interference  is  usually unavoidable in HetNets. Fortunately, according to  \cite{ Dhillon:2014dzba, Zhao:2015wpbaca}, the
	interference  can be re-utilized  as an effective  radio-frequency (RF)  energy source for wireless energy harvesting (WEH), which  thus contributes to  the green and self-sustainable  communications.  WEH has  many  advantages over  conventional energy supply methods \cite{Zhou:2013bjba}. For example,    WEH  is  more reliable than natural energy supply, such as solar, wind and tide, which are   significantly  affected  by climate  and  terrain.  Also,  it is   more cost-effective  compared to the  widely adopted  batteries recharge/replacement technique.  Generally, the densely deployed HetNets  are   favorable  from the perspective of  improving efficiency of WEH, since the distances from  energy harvesters to  energy stations are substantially shortened.

Recently,   WEH-based HetNets  have received  extension attention, in which the  power
	beacons (PBs)  provide
	power  for other  nodes  via  wireless energy transfer \cite{Tabassum:2015tsba,Lohani:2016etba,Akbar:2016htba,Zhu:2016caba,Sheng:2016icba,Zhang:2017kqba,Kim:2017ewba}. In \cite{Tabassum:2015tsba}, H.~Tabassum and E.~Hossain   studied the optimal deployment  of PBs  in wirelessly powered cellular networks.
	In \cite{Lohani:2016etba},  the downlink resource allocation problem was investigated by S.~Lohani et al. under  both time-switching and power-splitting   based  simultaneous wireless information
	and power transfer (SWIPT) strategies for two-tier HetNets. The  comprehensive  analysis of the  outage probability and the average ergodic rate in both downlink and uplink stages  of  wirelessly powered  HetNets with different cell  associations  were presented by S.~Akbar et al. in \cite{Akbar:2016htba}.  Y. Zhu et al. in 
	\cite{Zhu:2016caba}  further  extended   the above work into the Massive MIMO aided HetNets with WEH, where  different  user association schemes  are  investigated in terms of    the achievable average uplink
	rate.  From the view of green communications,   the
	energy efficient beamforming designs for SWIPT  HetNets   was  studied  by M. Sheng et al. and H. Zhang et al.  in \cite{Sheng:2016icba} and \cite{Zhang:2017kqba}, respectively.
	To  increase  energy harvesting efficiency  multi-antenna PBs and users,  J. Kim et al. in \cite{Kim:2017ewba}  considered sum
	throughput maximization   under different cooperative   protocols of two-tier
	wireless powered  cognitive
	networks.

	Furthermore, owing  to  the open network architectures of  HetNets, the security issue faced by wireless powered  HetNets has also drawn extensive  attention  \cite{Lv:2015knba}.   As a mature  technique  to  guarantee
 secure  communications  from the information-theoretical perspective, 	physical-layer security (PLS)  has  been widely researched in both academia and industry  \cite{Tekin:2008vbba, Shiu:2011ua}. There have been some works considering applying  PLS techniques  to HetNets with WEH. 
 In \cite{Ren:2017txba} and \cite{Li:2017waba}, the artificial  noise based  secrecy rate
	maximization was  studied  for  secure  HetNets with  SWIPT.
	The authors in \cite{Hu:2018kyba} proposed the max-min  secrecy energy efficiency
	optimization for   wireless powered HetNets, and    a distributed ADMM approach  was applied to
	reduce the  information
	exchange overhead. Considering the  more practical scenarios where  the transmitter only
	has imperfect eavesdropper's CSI,  a secrecy SWIPT strategy  for two-tier cognitive radio networks  was  investigated in  \cite{Ng:2016jzba}. 
	
	Most of the  existing
	works on  HetNets with WEH  focus
	on  SWIPT HetNets,  it is still an open challenging  issue on how to design the optimal harvest-then-transmit  strategies for HetNets. In this paper,  we investigate the secrecy beamforming design for  a wirelessly powered HetNet, where the  wirelessly powered FBS  transmits  the confidential information  to
	a  single-antenna  FU in the presence of  a multi-antenna eavesdropper (Eve). The FBS can harvest energy from  the  PB and  the  MBS. Moreover, there is no cooperation among the MBS, the PB and the FBS, thus the resultant cross-tier interference is 
	taken into account. In this wirelessly powered HetNet, the energy and information  covariance matrices as well as    the  time splitting  factor are jointly optimized to maximize the secrecy  rate under  different levels of  FBS-EVE CSI. 	 Our main  contribution is  summarized as :
	\begin{itemize}
		\item  Firstly, we  study  secrecy  rate maximization (SRM) of   the wirelessly powered  HetNet  having  perfect global CSI. In order to  address  this  non-convex perfect 
		SRM problem, a relaxed  problem using the matrix  trace inequality  is studied  and   proved to be   tight,  since it  always provides  a rank-1 optimal solution.  Considering  the  joint-convexity and quasi-convexity  of the  relaxed problem on   different  variables, a  convexity-based  linear search is proposed for  optimally  solving  the perfect SRM problem. In particular,  the  closed-form solution  of this problem is derived  when the cross-tier interference at the MU  is  negligible.
		
		\item  Secondly, the imperfect  FBS-EVE CSI with deterministic and  Gaussian random CSI errors  are   considered, respectively.  For the deterministic CSI error, the worst-case SRM  problem is studied, which can be addressed similarly to the  perfect SRM problem using the S-procedure. For the Gaussian random CSI error,  the  outage-constrained SRM problem subject to the probabilistic secrecy rate  constraint  is studied by applying the Bernstein-type inequality (BTI) \cite{Yuan:2018vzba} and  then  an alternating optimization procedure is proposed. For both the worst-case  and the outage-constrained   SRM  problems,  the rank-1  property  of the optimal  solutions is  also  validated.

		\item  Finally, we consider the realistic scenario with  unknown   FBS-EVE CSI, in  which   artificial noise (AN) is utilized for  improving   secrecy performance.   We design   AN aided secrecy beamforming  by maximizing  the average  AN power  subject to the legitimate  rate requirement at  the FU. This robust design  can  be  reformulated  as  a concave one and  its  optimal  rank-1 solution is demonstrated.  Similarly,  for the inactive cross-interference constraint,   the closed-form   solution to   this AN aided secrecy beamforming design is available.
	\end{itemize}
	
In fact,  the studied  SRM  problems  belong  to  the  nonconvex  difference of convex functions (DC) programming, which   are   more  challenging   than  that  in \cite{Kim:2017ewba} focusing on  the  sum-throughput maximization of  cognitive WPCNs.  Compared to the SRM  problem of \cite{Wu:2016foba} where   the  single-antenna PB and  Eve are  assumed,  these  SRM problems   are also  more intractable  due to the additional  energy and interference constraints. Fortunately,  we validate  that  the optimal  energy  and information  beamformings are of rank-1 in the secrecy wirelessly powered HetNet,
regardless of the availability of   eavesdropper's  CSI. This  conclusion also provides important insights  for   practical   engineering applications. 

\textit{Notations:} 
 The bold-faced lower-case and upper-case letters stand for vectors and matrices,
 respectively. The operators
 $(\cdot )^{\rm T}$, $(\cdot )^{\rm H}$ and $(\cdot )^{-1}$ denote the transpose, Hermitian and inverse of a matrix, respectively. 
 $\text{Tr}(\bm{A})$ and $\det (\bm{A})$ represent  the trace and determinant of
 $\bm{A}$, respectively.  $\|\cdot \|_2$ denotes the matrix
 spectral norm and $\bm{A}\!\succeq\! \bm{0}$ indicates that the square matrix
 $\bm{A}$ is positive semidefinite.   $\text{rank}(\bm{A})$ and   $\nu({\bf A})$ denote the rank of  ${\bf A}$ and  the unit-norm eigenvector associated with  the maximum eigenvalue of  ${\bf A}$, respectively. Also,  $(a)^{+}\!=\!\max\{a,0\}$ is defined.  The
 words `independent and identically distributed' and `with respect to' are
 abbreviated as `i.i.d.' and `w.r.t.', respectively.
 \vspace{-3mm}
\section{System model and Problem Formulation}
As shown in Fig.~\ref{Fig1},  we consider  secure communications of the wirelessly powered HetNet, in which an $N_M$-antenna MBS used for information transmission coexists with an $N_P$-antenna
 PB  deployed  for wireless energy transfer and  an $N_F$-antenna FBS  aiming for  energy harvesting.  Note that the FBS is
 \begin{figure}[t]
  \vspace{-8mm}
  \centering
  \includegraphics[width=3.3in]{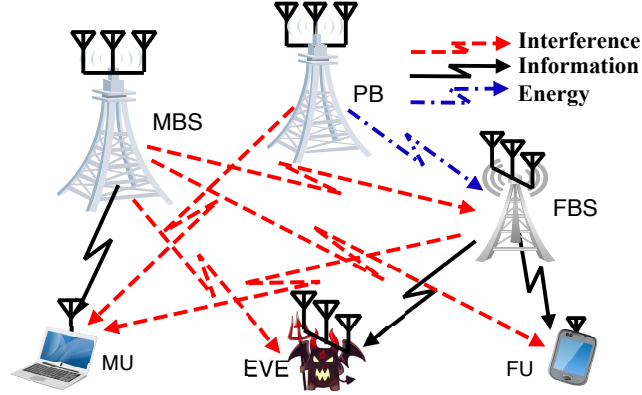}
  \caption{ A secrecy wirelessly powered HetNet.}
    \label{Fig1}
    \vspace{-8mm}
  \end{figure}
 energy-limited and  harvests energy for  its communications from the RF
signals transmitted by  the PB  and the MBS. The MBS and the PB  transmit the information-bearing signal ${ \bf   s}_{M}\in \mathbb{C}^{N_M}$ and energy-bearing signal ${ \bf  { s}}_{P}\in \mathbb{C}^{N_P}$  to a single-antenna MU and an $N_F$-antenna FBS, respectively. Then the  FBS transmits the confidential signal  ${ \bf   s}_{F}\in \mathbb{C}^{N_F}$  to a single-antenna FU, while   a  multi-antenna  Eve aims for intercepting  the signal of the FBS.  We focus our attention on the security  of  FBS  and  there is no cooperation
between   the MBS and the FBS.  Hence,  the signals transmitted
from the MBS and  the FBS  actually  impose    interference  on the FU and the MU,  respectively.  All  wireless channels are assumed to be quasi-static flat-fading   and  remain  constant during a whole  time slot $T$.

In the initial $\tau T$   subslot,   where  $0<\tau<1$  denotes the   time splitting factor, the FBS  harvests energy from both the PB energy signal ${ \bf   s}_{P}$  and  MBS  interfering signal ${ \bf   s}_{M}$.  Let's define     ${ \bf    W}_P=\mathbb{E}[{ \bf    s}_{P}{ \bf    s}_{P}^H]  \!\in\!\mathbb{C}^{N_P \times N_P}$ as  the  covariance matrix of the PB  energy signal  ${ \bf   s}_{P}  $   subject  to the maximum  transmit  power $P_P$, i.e.  $\text{tr}({ \bf    W}_P) \le P_P$, and  for simplicity 
   $\mathbb{E}[{ \bf   s}_{M}{ \bf   s}_{M}^H]=\frac{P_M}{N_M}{ \bf   I}_{N_M}$   with    $P_M$  being   
   the MBS maximum   transmit power.  By neglecting   the 
  contribution of thermal noise to the total  harvested energy  at the FBS,   the   amount of energy  harvested at  the  FBS  is expressed as
\begin{align}\label{EH1}
  E({ \bf    W}_P)=\tau T\xi\big(\text{tr}({ \bf    H}_F{ \bf    W}_P{ \bf    H}_F^H)+\frac{P_M}{N_M}\text{tr}(
  { \bf    G}_F{ \bf    G}_F^H)\big)
\end{align}
where $0<\xi< 1$  is the energy harvesting efficiency factor. ${ \bf H}_F\in\mathbb{C}^{ N_F \times N_P}$ and ${ \bf    G}_F
\in\mathbb{C}^{N_F \times N_M}$ denote the PB-FBS channel and  MBS-FBS channel, respectively.  

Next, in the  second $(1-\tau)T$  subslot, the FBS transmits the confidential  signal ${ \bf   s}_{F}$ to the FU by  utilizing  the harvested energy in \eqref{EH1}. The   signals  received at the FU and  the Eve are then  expressed  as
\begin{align}\label{RS1}
  &{ y}_R={ \bf h}_R^H{ \bf s}_F+{\bf g }_R^H{\bf s}_M+{ n}_R, \nonumber\\
&{ \bf  y}_E={ \bf H}_E{ \bf s}_F+{\bf G }_E{\bf s}_M +{ \bf   n}_E,
\end{align}
where ${ \bf h}_R\in\mathbb{C}^{ N_F }$ and ${ \bf g}_R
\in\mathbb{C}^{N_M}$ denote the FBS-FU channel and the MBS-FU channel, respectively.  ${ \bf H}_E\in\mathbb{C}^{ N_E\times N_F }$ and ${ \bf G}_E\in \mathbb{C}^{ N_E\times N_M}$ denote the FBS-EVE channel  and   the  MBS-EVE channel, respectively.
 $n_R\sim \mathcal{CN}({0,\sigma_n^2})$
 and ${ \bf n}_E \sim \mathcal{CN}({\bf 0}, \sigma_n^2{\bf I}_{N_E})$  are    i.i.d  circularly symmetric  Gaussian noises at the FU and the EVE, respectively. Additionally, we define 
  ${ \bf  P}_F=\mathbb{E}[{ \bf    s}_{F}{ \bf    s}_{F}^H]  \!\in\!\mathbb{C}^{N_F \times N_F}$  as    the covariance matrix of the FBS information  signal  ${ \bf   s}_{F}  $, the  achievable  rates (in  bps/Hz) at the FU and the EVE  are then given by 
\begin{align}\label{RS1}
  & R_I= (1-\tau)\log_2\bigg(1+ \frac{ { \bf   h}_{R}^H
   { \bf   P}_{F}
   { \bf   h}_{R}}{\sigma_n^{2}+\frac{P_M}{N_M}\Vert {\bf g }_R\Vert^2}\bigg),\nonumber\\
& R_E= (1-\tau)\log_2\det\big({{ \bf   I}}_{N_E}\!+\! (\sigma_n^{2}{ \bf   I}_{N_E}+\frac{P_M}{N_M}
   { \bf   G}_E{\bf G }_E^H)^{-1} \!{ \bf   H}_{E}
   { \bf   P}_{F}{ \bf   H}_{E}^H\big).
\end{align}

 According  to  \cite{Ren:2017txba}, the achievable secrecy rate $R_S$ of the  wirelessly powered HetNet is  actually the   data  rate at which   the   desired information  is  correctly decoded by the  FU, while  no  information  is  wiretapped by the EVE. Mathematically, we  have
\begin{align}\label{RS1}
R_S=[R_I-R_E]^+.
\end{align}

In our work, we jointly optimize the   PB and  FBS transmit covariance matrices $\{{ \bf W}_P, { \bf P}_F\}$  and the time splitting factor $\tau$  for maximizing the achievable secrecy rate  $R_S$ in \eqref{RS1}.  The resultant SRM problem  for the  wirelessly powered HetNet is then  formulated as 
\begin{align}\label{SRM}
  &{R}_S^{\star}= \underset{{{\tau},{ \bf   W}_P\succeq{ \bf   0}, { \bf   P}_{F}
  \succeq{ \bf   0} }}
{\text{max}}~~R_I-R_E\nonumber\\
&~{\rm{s.t.}}~~\text{CR1:}~0\le\tau\le 1,~\text{tr}({ \bf   W}_P) \le P_P\nonumber\\
 &~~~~~~~\text{CR2:}~(1\!-\!\tau)\text{tr}({ \bf   P}_{F})\!
 \le \!\tau \xi \big[\text{tr}({ \bf    H}_F{ \bf    W}_P{ \bf    H}_F^H)+\frac{P_M}{N_M}\text{tr}(
  { \bf    G}_F{ \bf    G}_F^H)\big],\nonumber\\
 &~~~~~~~\text{CR3:}~\tau{ \bf   h}_{M}^H { \bf   W}_P{ \bf   h}_{M} +(1\!-\!\tau){ \bf g}_{p}^H { \bf P}_{F}{ \bf g}_{p}\!
 \le \! I_{th}.
\end{align}

In  problem (\ref{SRM}), the constraint $\text{CR1}$  comes from  the fact that the  PB transmit power has  a maximum threshold, and  $\text{CR2}$  denotes  the    energy causality constraint   of the wirelessly powered  FBS.  While $\text{CR3}$   models  the average interference constraint   of the secrecy wirelessly powered HetNet. Specifically, by defining   ${ \bf h}_{M}\in \mathbb{C}^{ N_P}$ and ${ \bf g}_{p}\in \mathbb{C}^{ N_F}$  as   the  PB-MU channel and the FBS-MU channel, respectively,
 the terms ${ \bf h}_{M}^H { \bf   W}_P{ \bf   h}_{M}$ and ${ \bf g}_{p}^H { \bf P}_{F}{ \bf g}_{p}$ actually  denote the  total interference  at  the MU  originating  from  the PB and the FBS, respectively.
Since the PB energy transfer and the FBS information transfer are separated by  the time splitting factor $\tau$,  we consider  the average interference power constraint  at the MU  as shown in $\text{CR3}$, where $ I_{th}$ denotes the  minimum  tolerable interference.  It is readily inferred  from $\text{CR3}$ that  problem \eqref{SRM} is feasible for an  arbitrary  $I_{th}\ge 0$. However, due to the  highly coupled variables $\{{\tau},{\bf W}_P,{\bf P}_{F}\}$, the  SRM problem \eqref{SRM}  is generally non-convex and   challenging to address.

In the sequel, we will investigate the SRM  under  three  different levels of the FBS-EVE CSI. In the first case,  the  global CSI  of the wirelessly powered  HetNet   is available at the FBS via   channel feedback and  high-SNR training techniques. In the second case,  by assuming imperfect FBS-EVE  CSI at the FBS,  a pair of robust SRM  problems   are  investigated  under  deterministic  and Gaussian random CSI errors, respectively.
In the third case,  we consider the more practical scenario that the
FBS  is not aware  of the existence of Eve. In other words, the
FBS-EVE  CSI is  completely unknown to the FBS.
\vspace{-3mm}

 \section{Perfect SRM under global   CSI of  the wirelessly powered  HetNet }\label{III}
  In this section, a  convexity-based one  dimensional search  is proposed  for optimally solving the  SRM problem  \eqref{SRM}  under the global  CSI of  the    wirelessly powered HetNet.
  \vspace{-4mm}
  \subsection{Transformation of Problem \eqref{SRM}}
  Firstly,  by introducing an auxiliary variable $\eta$, the   SRM  problem \eqref{SRM} can be  reformulated as
  \begin{align}\label{SRM2}
  &{R}_S^{\star}=\underset{\substack{{\tau},{ \bf W}_P\succeq{ \bf   0},
  { \bf P}_{F}\succeq{ \bf   0},\eta }}
{\text{max}}~(1-\tau)\log_2\bigg(1+ { C_1{ \bf   h}_{R}^H
   { \bf   P}_{F}
   { \bf   h}_{R}}\bigg)\!-\!(1\!-\!\tau)\log_2\eta\nonumber\\
&{\rm{s.t.}}~~\text{CR1}\sim\text{CR3},~~~\text{CR4:}~
   \det\big({{ \bf   I}}_{N_E}\!+\! {\bf R}_E { \bf   H}_{E}
   { \bf   P}_{F}{ \bf   H}_{E}^H\big)\le \eta,
\end{align} where $C_1={(\sigma_n^{2}+\frac{P_M}{N_M}\Vert {\bf g }_R\Vert^2)^{-1}}$ and ${\bf R}_E=(\sigma_n^{2}{ \bf   I}_{N_E}+\frac{P_M}{N_M}
   { \bf   G}_E{\bf G }_E^H)^{-1}$. Unfortunately,  problem   \eqref{SRM2} is still difficult  to address   because of the nonconvex  constraint $\text{CR4}$. In order to tackle this issue,  we  firstly consider a relaxation  of  $\text{CR4}$ based on  the following lemma.
   \vspace{-2mm}
 \begin{lemma}\label{lemm0}\cite{Li:2011ivba}
For any positive semi-definite matrix ${\bf A}\succeq {\bf 0}$,  we have
  $\det({\bf I} + {\bf A}) \ge 1+ \text{tr}({\bf A})$,
where the equality  holds if and only if  $\text{rank}({\bf A})=1$.
  \end{lemma}   \vspace{-3mm}
By applying Lemma~\ref{lemm0} to   $\text{CR4}$,  problem  \eqref{SRM2}  is relaxed to
 \begin{align}\label{SRMRe}
   &\tilde{R}_S^{\star}=\underset{\substack{{\tau},{ \bf W}_P\succeq{ \bf   0},
  { \bf P}_{F}\succeq{ \bf   0},\eta }}
{\text{max}}~(1-\tau)\log_2\bigg(1+ { C_1{ \bf   h}_{R}^H
   { \bf   P}_{F}
   { \bf   h}_{R}}\bigg)\!-\!(1\!-\!\tau)\log_2\eta\nonumber\\
&~~~~~{\rm{s.t.}}~~~\text{CR1}\sim\text{CR3},~~\text{CR4:}~
   1\!+\! \text{tr}\big({\bf R}_E { \bf   H}_{E}
   { \bf   P}_{F}{ \bf   H}_{E}^H\big)\le \eta,
 \end{align} where $\tilde{R}_S^{\star}$ denotes the  achievable  maximum  secrecy rate by solving    problem \eqref{SRMRe}.
Based on  Lemma~\ref{lemm0},  problem \eqref{SRMRe} clearly   has a larger feasible  region   than   problem \eqref{SRM2}, so that  $\tilde{R}_S^{\star}\ge {R}_S^{\star}$ holds.  We then  introduce   new
variables  $\widetilde{\bf W}_P=
\tau{\bf W}_P$ and $\widetilde{\bf P }_{F}
=(1\!-\!\tau) {\bf P}_{F}$ to equivalently transform   problem  \eqref{SRMRe} into
\begin{align}\label{SRMRe3}
   &\tilde{R}_S^{\star}=\underset{\substack{{\tau}, \widetilde{ \bf W}_P\succeq{ \bf   0}, 
  \widetilde{ \bf P}_{F}\succeq{ \bf   0},  \eta }}
{\text{max}}~(1-\tau)\log_2\bigg(\frac{1+ \frac{ C_1{ \bf   h}_{R}^H
   \widetilde{ \bf   P}_{F}
   { \bf   h}_{R}}{1-\tau}}{\eta}\bigg)\nonumber\\
&~~~{\rm{s.t.}}~~\widetilde{\text{CR1}}:~0\le\!\tau\!\le 1,~\text{tr}(\widetilde{ \bf   W}_P)\! \le \!\tau P_P, ~~\widetilde{\text{CR2}}:\text{tr}(\widetilde{ \bf   P}_{F})\!
 \le \! \xi \text{tr}({ \bf    H}_F\widetilde{ \bf    W}_P{ \bf    H}_F^H)\!+\!\xi \tau\frac{P_M}{N_M}\text{tr}(
  { \bf    G}_F{ \bf    G}_F^H),\nonumber\\
 &~~~~~~~~~\widetilde{\text{CR3}}:{ \bf   h}_{M}^H \widetilde{ \bf   W}_P{ \bf   h}_{M} \!+\!{ \bf g}_{p}^H \widetilde{ \bf P}_{F}{ \bf g}_{p}\!
 \le \! I_{th}, ~~\widetilde{\text{CR4}}:
 1\!-\!\tau\!+\! \text{tr}\big( {\bf R}_E{ \bf   H}_{E}
  \widetilde { \bf   P}_{F}{ \bf   H}_{E}^H\big)\le \eta(1\!-\!\tau).
 \end{align}
 
 It is concluded  from   problem  \eqref{SRMRe3} that for any fixed $\eta$  the objective function  is the perspective  of a strictly concave  matrix
  function  $f( \widetilde{\bf P}_{F}) \!= \!
\log_2\bigg(\frac{1+ { C_1{ \bf   h}_{R}^H
   \widetilde{ \bf   P}_{F}
   { \bf   h}_{R}}}{\eta}\bigg)$, which is also strictly concave   \cite[p. 39]{Step:2018twba}. Moreover, all
constraints in  problem  \eqref{SRMRe3} are  convex.  Therefore, it is  inferred that  problem  \eqref{SRMRe3} is  jointly  concave w.r.t.
$\{{\tau}, \widetilde{\bf W}_P,
\widetilde{\bf P}_{F}\}$ given any $\eta$, and can be globally solved  by  the interior point method.
To further  investigate the tightness of the  constraint $\widetilde{\text{CR3}}$  when varying  $I_{th}$,  we  firstly consider  solving  problem \eqref{SRMRe3} without  $\widetilde{\text{CR3}}$
   \begin{align}\label{SRMRe31}
   &\underset{\substack{{\tau},\widetilde{ \bf W}_P\succeq{ \bf   0},
  \widetilde{ \bf P}_{F}\succeq{ \bf   0},\eta }}
{\text{max}}~(1-\tau)\log_2\bigg(\frac{1+ \frac{ C_1{ \bf   h}_{R}^H
   \widetilde{ \bf   P}_{F}
   { \bf   h}_{R}}{1-\tau}}{\eta}\bigg), ~~~{\rm{s.t.}}~~\widetilde{\text{CR1}},~\widetilde{\text{CR2}},
~\widetilde{\text{CR4}}.
 \end{align}   
 
 Let's define $\widetilde{ \bf   W}_{P, I_{th}}$ and $\widetilde{ \bf P}_{F,I_{th}}$  as  the optimal solution to   problem \eqref{SRMRe31}.  The corresponding interference level is then expressed  as
 $\tilde{I}_{th}= { \bf   h}_{M}^H \widetilde{ \bf   W}_{P,I_{th}}{ \bf   h}_{M} +{ \bf g}_{p}^H \widetilde{ \bf P}_{F,I_{th}}{ \bf g}_{p}$. When ${I}_{th}>\tilde{I}_{th}$,  it is  readily inferred  that the constraint $\widetilde{\text{CR3}}$ in  problem  \eqref{SRMRe3}  will be automatically satisfied using  the  optimal solution $\{\widetilde{ \bf   W}_{P, I_{th}}, \widetilde{ \bf P}_{F,I_{th}}\}$ to   problem \eqref{SRMRe31}, implying  that  in this context  $\widetilde{\text{CR3}}$  has no effect on   problem  \eqref{SRMRe3} and  can be neglected  without loss of optimality.   While for the case of   $0\le I_{th}\le \tilde{I}_{th}$,  we   provide an interesting insight  in the following  Theorem.
\vspace{-3mm}
 \begin{proposition}\label{thre1}
When $0\le I_{th}\le \tilde{I}_{th}$,  the constraint $\widetilde{\text{CR3}}$ in  problem \eqref{SRMRe3}  is tight, which   implies   that  the optimal solution to   problem \eqref{SRMRe3} lies at  the boundary of ${ \bf   h}_{M}^H \widetilde{ \bf   W}_P{ \bf   h}_{M} +{ \bf g}_{p}^H \widetilde{ \bf P}_{F}{ \bf g}_{p}\!
 = \! I_{th}$.
\end{proposition}
\vspace{-5mm}
\begin{proof}
Theorem~\ref{thre1}  is proved by  contradiction  as follows.  Firstly, we  consider  an interference threshold  $I_{th}^1$ satisfying $0\le I_{th}^1\le  \tilde{I}_{th}$ and  denote the obtained    $\tilde{R}_S^{\star}$ by solving  problem  \eqref{SRMRe3} as $\tilde{R}_S^{\star}=f_{obj, I_{th}^1}^{\star}({\mathcal{Q}}_1)$, where ${\mathcal{Q}}_1=\{\widetilde{ \bf P}_{F,1}^{\star},\widetilde{ \bf W}_{P,1}^{\star},\tau_1^{\star}, \eta_1^{\star}\}$  is the optimal solution to  problem   \eqref{SRMRe3} with  $I_{th}=I_{th}^1$ and  ${ \bf   h}_{M}^H \widetilde{ \bf   W}_{P,1}^{\star}{ \bf   h}_{M} +{ \bf g}_{p}^H \widetilde{ \bf P}_{F,1}^{\star}{ \bf g}_{p}\!< I_{th}^1$ is assumed. It is then  easily found   that there is  another  interference
 threshold $I_{th}^2$ satisfying  $I_{th}^2={ \bf   h}_{M}^H \widetilde{ \bf   W}_{P,1}^{\star}{ \bf   h}_{M} +{ \bf g}_{p}^H \widetilde{ \bf P}_{F,1}^{\star}{ \bf g}_{p}\!< I_{th}^1$, based on which   the  optimal solution   ${\mathcal{Q}}_1$  actually becomes a  feasible solution to  problem \eqref{SRMRe3}.  In other words, we  have
   $f_{obj,I_{th}^1}^{\star}({\mathcal{Q}}_1) \!\le \! f_{obj, I_{th}^2}^{\star}({\mathcal{Q}}_2)$,
where ${\mathcal{Q}}_2\!=\!\{\widetilde{ \bf P}_{F,2}^{\star},\widetilde{ \bf W}_{P,2}^{\star},\tau_2^{\star}, \eta_2^{\star}\}$  is the  optimal solution to  problem \eqref{SRMRe3} with $I_{th}=I_{th}^2$. On the other hand, since  $I_{th}^2<I_{th}^1$, a  smaller feasible  region   is observed  for problem \eqref{SRMRe3}  with  $I_{th}=I_{th}^2$,  which   thus yields 
  $ f_{obj,I_{th}^1}^{\star}({\mathcal{Q}}_1)\! \ge \! f_{obj,I_{th}^2}^{\star}( {\mathcal{Q}}_2).$
   By combining the above two inequalities, it is concluded  that
   $f_{obj,I_{th}^1}^{\star}({\mathcal{Q}}_1)=
    f_{obj,I_{th}^2}^{\star}({\mathcal{Q}}_2).$
 Similarly, for an arbitrary threshold $I_{th}\in[I_{th}^2,I_{th}^1]$, the same   maximum objective value of  problem \eqref{SRMRe3}  is observed. This phenomenon   hints  that the  constraint $\widetilde{\text{CR3}}$  actually has no  effect on  problem \eqref{SRMRe3} and thus can be ignored without loss of optimality. As discussed before, this happens  only when $ {I}_{th}>\tilde{I}_{th}$,  which  contradicts to the original condition  $0\le I_{th}\le \tilde{I}_{th}$. Therefore,  the initial assumption  of ${ \bf   h}_{M}^H \widetilde{ \bf   W}_{P,1}^{\star}{ \bf   h}_{M} +{ \bf g}_{p}^H \widetilde{ \bf P}_{F,1}^{\star}{ \bf g}_{p}\!< I_{th}^1$ is actually  invalid,  and   we must have the optimal solution $\{\widetilde{ \bf   W}_{P,1}^{\star},\widetilde{ \bf P}_{F,1}^{\star}\}$  at the boundary of ${ \bf   h}_{M}^H \widetilde{ \bf   W}_{P,1}^{\star}{ \bf   h}_{M} +{ \bf g}_{p}^H \widetilde{ \bf P}_{F,1}^{\star}{ \bf g}_{p}\!= I_{th}^1$ for  problem \eqref{SRMRe3}  when $0\le I_{th}^1 \le \tilde{I}_{th}$.
\end{proof}
Since the globally optimal solution $\{{\tau}, \widetilde{\bf W}_P,
\widetilde{\bf P}_{F}\}$ to  problem \eqref{SRMRe3}  for any fixed $\eta$ is   available,  our remaining  task is to find the globally optimal $\eta$, as presented  in Theorem~\ref{thre2}. 
\vspace{-2mm}
 \begin{proposition}\label{thre2}
Problem  \eqref{SRMRe3}   is
 quasi-concave  w.r.t.  $\eta$, and the globally optimal $\eta$  can be computed by a one-dimensional search, i.e., Golden section search (GSS) \cite{Step:2018twba}.
\end{proposition}
\vspace{-5mm}
  \begin{proof}
Let's rewrite  problem \eqref{SRMRe3} by introducing an auxiliary variable $t$ as
  \begin{subequations}\label{ARM1}
  \begin{align}
   &\underset{\substack{{\tau}, \eta, t,\widetilde{ \bf W}_P\succeq{ \bf   0},
  \widetilde{ \bf P}_{F}\succeq{ \bf   0}}}
{\text{max}}~(1-\tau)\log_2\bigg(1+ \frac{C_1 { \bf   h}_{R}^H
  \widetilde{ \bf   P}_{F}
   { \bf   h}_{R}}{1-\tau}\bigg)\!+\! t \label{ARM10}\\
&~~~~~~{\rm{s.t.}}~~\widetilde{\text{CR1}}\sim \widetilde{\text{CR3}},\label{ARM11} \\
&~~~~~~~~~~~~\widetilde{\text{CR4}}:~
 1\!-\!\tau\!+\! \text{tr}\big( {\bf R}_E{ \bf   H}_{E}
  \widetilde { \bf   P}_{F}{ \bf   H}_{E}^H\big)\le \eta(1\!-\!\tau) \label{ARM12}\\
&~~~~~~~~~~~~\widetilde{\text{CR5}}: t\le (1\!-\!\tau)\log_2\frac{1}{\eta}. \label{ARM13}
 \end{align}
\end{subequations}

Given any $\eta$,  problem \eqref{ARM1} is  jointly and  strictly concave w.r.t. $\{{\tau},t,\widetilde{ \bf W}_P,
  \widetilde{ \bf P}_{F}\}$ and thus the  unqiue optimal solution  exists. Based on this, it is readily  inferred that the value of \eqref{ARM10}  is continuous on $\eta$.  We  observe that
for a sufficiently  small $\eta$, the value of \eqref{ARM10} is strongly dominated  by the active constraint \eqref{ARM12}.  Upon increasing  $\eta$, the
feasible  region specified by \eqref{ARM12} expands and thus   the value of \eqref{ARM10} increases.  However, when $\eta$ becomes large enough, the constraint  \eqref{ARM13}   with  the small $\log_2\frac{1}{\eta}$ actually  dominates   the value of  \eqref{ARM10}.  In this context,   we   find  that the value of \eqref{ARM10}  decreases with   increasing  $\eta $.  According to the above analysis, it can be   inferred  that there must exist  a turning point  $\hat{\eta}$ for  problem \eqref{SRMRe3}.  Specifically,   with the  increase of  small  $\eta$,  the   value of  \eqref{ARM10}  firstly  increases   until  $\eta$ reaches $\hat{\eta}$, and  then  decreases. This property hints that  problem  \eqref{SRMRe3} is
unimodal (quasi-concave)  w.r.t. $\eta$, of which the globally optimal  value   can be  found by  GSS.
  \end{proof}
 According to Theorem~\ref{thre2},  we firstly determine
    the one-dimensional  search interval  of $\eta$ as follows.  For achieving    a nonzero  secrecy rate,   the maximum value of $\eta$ actually corresponds to the maximum legitimate rate  $R_{up}$ of the FU, which is derived   by solving the following problem
\begin{align}\label{RM1}
   &R_{up}=\underset{\substack{{\tau},\widetilde{ \bf W}_P\succeq{ \bf   0},
  \widetilde{ \bf P}_{F}\succeq{ \bf   0}}}
{\text{max}}~(1-\tau)\log_2\bigg(1+ \frac{ C_1{ \bf   h}_{R}^H
   \widetilde{ \bf   P}_{F}
   { \bf   h}_{R}}{1-\tau}\bigg),~~~{\rm{s.t.}}~~\widetilde{\text{CR1}}\sim \widetilde{\text{CR3}}.
 \end{align} 
 
It is clear  that  problem \eqref{RM1} is also jointly concave w.r.t. $\{{\tau},\widetilde{ \bf W}_P,$ $
  \widetilde{ \bf P}_{F} \}$.  With the  optimal solution $\{R_{up}^{\star},\tau_1^{\star} \}$ to  problem \eqref{RM1},  we  readily infer  that  the value range of    $\eta$  is  $1\le\eta\le 2^{\frac{R_{up}^{\star}}{1-\tau_1^{\star}}}$. Overall, given any $\eta$,  by combining the  joint concavity of  problem  \eqref{SRMRe3}  with  the    GSS for  finding the globally-optimal  $\eta$,   the problem \eqref{SRMRe3}   can be optimally  solved.  More importantly,  we further prove that  ${R}_S^{\star}=\tilde{R}_S^{\star}$   holds  for the  PSRM  problem \eqref{SRM}  and the relaxed problem \eqref{SRMRe3}, as shown  in  Theorem~\ref{thre3}.
  \vspace{-3mm}
\begin{proposition}\label{thre3}
Problem \eqref{SRMRe3}  is a tight relaxation  of  the perfect SRM problem  \eqref{SRM}.  In other words, the optimal solution $\{{\tau}^{\star},\eta^{\star},\widetilde{ \bf W}_P^{\star},
  \widetilde{ \bf P}_{F}^{\star} \}$ to  problem  \eqref{SRMRe3} satisfies  $\text{rank}(\widetilde{ \bf W}_P^{\star})\!\!=\!\!\text{rank}(\widetilde{ \bf P}_{F}^{\star} )\!\!=\!\!1$ and thus  $\tilde{R}_S^{\star}\!=\!{R}_S^{\star}$ holds.  By recalling $\widetilde{\bf W}_P=
\tau{\bf W}_P$ and $\widetilde{\bf P }_{F}
=(1\!-\!\tau) {\bf P}_{F}$, the corresponding  $\{{\tau}^{\star},{ \bf W}_P^{\star},
  { \bf P}_{F}^{\star} \}$  are  also  globally-optimal  to problem \eqref{SRM}.
\end{proposition}
\vspace{-5mm}
  \begin{proof}
    Please refer to Appendix~\ref{AppenA}.
  \end{proof}
\vspace{-3mm}
 
Following   the proof of Theorem~\ref{thre3},  we propose 
  a convexity-based  linear search  for globally solving the  perfect  SRM problem \eqref{SRM}.  To be specific,  given any  $\eta$,   the  relaxed  convex problem   \eqref{SRMRe3}   is  firstly addressed   for obtaining  the optimal  solution  $\{{\tau}^{\star},{ \bf W}_P^{\star},
  { \bf P}_{F}^{\star}\}$ (via variable substitution)  and the resultant achievable    secrecy rate. Then  the  GSS is  applied to   find the globally-optimal $\eta^{\star}$.  Theorem~\ref{thre3}   reveals  that the optimal energy  and information covariance matrices, i.e. ${ \bf W}_P^{\star}$ and $
  { \bf P}_{F}^{\star}$,   to problem \eqref{SRM} satisfy $ \text{rank}({ \bf W}_P^{\star})\!=\!\text{rank}({ \bf P}_{F}^{\star} )\!=\!1$. Physically, this property   means   that  single-stream transmission of  both   PB and  FBS  are   optimal for secrecy performance  of the  wirelessly powered HetNet.
\vspace{-3mm}
\section{Robust  SRMs   under  Imperfect FBS-EVE  CSI }
In this section,  two types of robust SRM problems  are investigated in depth  for the secrecy wirelessly powered HetNet. Specifically, one is the worst-case SRM  associated with deterministic FBS-EVE CSI error. In this case, we   propose  a convexity-based linear search method  for finding the  optimal  worst-case  solution. The other is the outage-constrained  SRM  subject to  Gaussian random FBS-EVE CSI error, for which the   convex reformulation  is  realized  by introducing an auxiliary variable.
\vspace{-4mm}
\subsection{The Proposed  Worst-Case SRM}\label{IV_A}
Recall the system model  in Section II, when considering deterministically  imperfect FBS-EVE channel, we have  
   ${ \bf   H}_{E}=\widehat{\bf H}_{E}+\Delta{ \bf   H}_{E}$,
  where $\widehat{{ \bf   H}}_{E_k}$  denotes  the estimated  FBS-EVE  channel and  $\Delta{ \bf   H}_{E}$  is the norm bounded CSI error, i.e,  $\Vert\Delta{ \bf   H}_{E}\Vert_F\le \xi_f$.  Based on this, the achievable worst-case secrecy rate of the wirelessly powered HetNet is  given by 
$R_{S,Ro}=R_I- \underset{{  \Delta}{ \bf   H}_{E}}{\max} ~ R_E$, and the  resultant  worst-case SRM problem is  formulated as
\begin{align}\label{WSRM}
  {R}_{S,Ro}^{\star}= &\underset{{{\tau},{ \bf   W}_P\succeq{ \bf   0}, { \bf   P}_{F}
  \succeq{ \bf   0} }}{\text{max}}~~\underset{{  \Delta}{ \bf   H}_{E}}{\text{min}}
~~R_I-  R_E,~~{\rm{s.t.}}~~\text{CR1}\sim \text{CR3}.
\end{align}
 As with   the reformulation of  perfect SRM problem \eqref{SRM},  we  also  introduce the auxiliary  variable  $\eta$ to reformulate problem \eqref{WSRM}   as
\begin{align}\label{WSRM2}
  &{R}_{S,Ro}^{\star}=\underset{{{\tau},{ \bf W}_P\succeq{ \bf   0},
  { \bf P}_{F}\succeq{ \bf   0},\eta }}
{\text{max}}~~(1-\tau)\log_2\bigg(1+ {C_1 { \bf   h}_{R}^H
   { \bf   P}_{F}
   { \bf   h}_{R}}\bigg)\!-\!(1\!-\!\tau)\log_2\eta\nonumber\\
&{\rm{s.t.}}~~\text{CR1}\sim\text{CR3}, ~~\text{CR5:}~
   \det\big({{ \bf   I}}_{N_E}\!+\! \bm{R}_E { \bf   H}_{E}
   { \bf   P}_{F}{ \bf   H}_{E}^H\big)\le \eta, ~\Vert\Delta{ \bf   H}_{E}\Vert_F\le \xi_f.
\end{align}

Compared to the   perfect SRM problem \eqref{SRM},   problem \eqref{WSRM2} is more challenging  since the semi-infinite norm bounded CSI error is included  in  $\text{CR5}$. To make    problem  \eqref{WSRM2}  tractable,  we firstly  relax   it using the above  Lemma~\ref{lemm0}   to 
\begin{align}\label{WSRM3}
  &\tilde{R}_{S,Ro}^{\star}=\underset{{{\tau},{ \bf W}_P\succeq{ \bf   0},
  { \bf P}_{F}\succeq{ \bf   0},\eta }}
{\text{max}}~~(1-\tau)\log_2\bigg(1+ {C_1 { \bf   h}_{R}^H
   { \bf   P}_{F}
   { \bf   h}_{R}}\bigg)\!-\!(1\!-\!\tau)\log_2\eta\nonumber\\
&{\rm{s.t.}}~~\text{CR1}\sim\text{CR3},~~\text{CR5:}~
   1\!+\! \text{tr}\big({\bf R}_E { \bf   H}_{E}
   { \bf   P}_{F}{ \bf   H}_{E}^H\big)\le \eta, ~\Vert\Delta{ \bf   H}_{E}\Vert_F\le \xi_f,
\end{align} where $\tilde{R}_{S,Ro}^{\star}$ denotes the maximum achievable  worst-case  secrecy rate by solving 
problem \eqref{WSRM3} and satisfies $\tilde{R}_{S,Ro}^{\star}\ge {R}_{S,Ro}^{\star}$ due to  larger feasible  region of problem \eqref{WSRM3}.  For solving  this non-convex problem  effectively,   both  the equality ${\bf h}_E=\text{vec}((\widehat{{ \bf   H}}_E+{\Delta}{{ \bf   H}_E})^H)
=\hat{{ \bf   h}}_E+{ \bf   \Delta}{{ \bf   h}_E}$ and the identity $\text{Tr}({\bf A}^H{\bf B}{\bf C}{\bf D})= \text{vec}({\bf A})^H({\bf D}^T\otimes {\bf B})\text{vec}({\bf C})$ are   utilized to rewrite  the constraint   $\text{CR5}$ as
  \begin{align}\label{eqrobus}
1+(\hat{{ \bf   h}}_E+{    \Delta}{{ \bf   h}_E})^H
   ({ \bf   R}_E\otimes { \bf   P}_F)(\hat{{ \bf   h}}_E+{   \Delta}{{ \bf   h}_E})\le \eta, ~\Vert\Delta{ \bf   h}_{E}\Vert_F\le \zeta_f.
   \end{align}
   \begin{lemma}\label{lemm1}  \vspace{-3mm}
{(S-procedure) \cite{Luo:2004ffba}}
For the equation $f({ \bf   x})={ \bf   x}^H{ \bf   A}{ \bf   x}+{ \bf   x}^H{ \bf   b}+
{ \bf   b}^H{ \bf   x}+c$, in which ${\bf  x}\in \mathbb{C}^{N}$, ${ \bf   A}\in
\mathbb{H}^{N\times N}$, ${ \bf   b}\in \mathbb{C}^{N}$ and $c$ is a constant, the  following equality  holds
\begin{align}
 f({ \bf   x})\! \le \! 0, \forall {\bf   x}\! \in \! \{{ \bf   x} |  \text{tr}({ \bf   x}{ \bf   x}^H)\! \le \! \! \epsilon_e\} \! \Leftrightarrow  \!  u\left[\begin{array}{cc}
{ \bf   I}_{N} &{ \bf   0}_{N\times 1} \\{ \bf   0}_{N\times 1}^T & \! \! -\!\!  \epsilon_e
\end{array}\right]\! \! -\! \! \left[\begin{array}{cc}
{ \bf   A} &{ \bf   b} \\{ \bf   b}^H& c
\end{array}\right]\! \! \succeq\! \! { \bf   0},  ~\text{with some} ~ u\! \ge\!  0.   \end{align}
\end{lemma}
We then take advantage of     Lemma~\ref{lemm1}  and   auxiliary variables $\widetilde{\bf W}_P=
\tau{\bf W}_P$, $\widetilde{\bf P }_{F}
=(1\!-\!\tau) {\bf P}_{F}$      to  transform  the nonlinear semi-infinite constraint  \eqref{eqrobus}  into a LMI. Then the relaxed WSRM problem \eqref{WSRM3}  is re-expressed   as
\vspace{-3mm}
   \begin{align}\label{WSRM5}
  &\tilde{R}_{S,Ro}^{\star}=\underset{\substack{{\tau},u, \eta,\\ \widetilde{ \bf W}_P\succeq{ \bf   0},
 \widetilde{ \bf P}_{F}\succeq{ \bf   0}}}
{\text{max}}~~(1-\tau)\log_2\bigg(1+\frac {C_1 { \bf   h}_{R}^H\widetilde{ \bf   P}_{F}
   { \bf   h}_{R}}{1-\tau}\bigg)\!-\!(1\!-\!\tau)\log_2\eta \\
&{\rm{s.t.}}~~\widetilde{\text{CR1}}\!\sim \!\widetilde{\text{CR3}},~\widetilde{\text{CR5}}:
 \!  \left[\begin{array}{cc}
u{ \bf   I}_{N}\!-\!({ \bf   R}_E^{ T}\otimes \widetilde{ \bf   P}_F) & -({ \bf   R}_E^{ T}\!\otimes\! \widetilde{ \bf   P}_F)\hat{ \bf   h}_E\\ -\hat{ \bf   h}_E^H({ \bf   R}_E^{ T}\!\otimes\! \widetilde{ \bf   P}_F)  & (\eta\!-\!1)(1\!-\!\tau)
-u\zeta_f^2-\hat{ \bf   h}_E^H({ \bf   R}_E^{ T}\otimes \widetilde{ \bf   P}_F)\hat{ \bf   h}_E
\end{array}\right]\!\succeq\!{ \bf   0},\nonumber
\end{align} where a scalar  variable $u>0$ is  introduced.  In contrast to problem \eqref{SRMRe3}, an additional SDP constraint $\widetilde{\text{CR5}}$ is included in 
  problem \eqref{WSRM5}.
 For any fixed $\eta$,  we easily find that $\widetilde{\text{CR5}}$  is  a convex linear matrix inequality (LMI) w.r.t. $\{{\tau},u,
  \widetilde{ \bf P}_{F}\}$,  so     problem \eqref{WSRM5} becomes  jointly concave w.r.t. $\{{\tau},u, \widetilde{ \bf W}_P,
  \widetilde{ \bf P}_{F} \}$.  Similar to   Theorem~\ref{thre2},  we can
    also   prove   that    problem \eqref{WSRM5} is  quasi-concave w.r.t $\eta$,   since  the    semi-infinite  CSI error ${\Delta}{{ \bf   h}_E}$  in the equivalent problem \eqref{WSRM3} is  independent of  $\eta$.  Based on  the above discussions, it is readily inferred that the proposed  convexity-based linear search  for   problem  \eqref{SRMRe3}  can be directly extended to   problem  \eqref{WSRM5}.
More importantly,  an interesting  insight is  provided in the following  Theorem, namely, 
  problem \eqref{WSRM5}  also admits the rank-1 optimal solution.
    \vspace{-3mm}
\begin{proposition}\label{thre4}
Problem \eqref{WSRM5}  is a tight relaxation of the worst-case SRM problem  \eqref{WSRM},  which  means that  its  optimal solution $\{{\tau}^{\star},\eta^{\star},\widetilde{ \bf W}_P^{\star},
  \widetilde{ \bf P}_{F}^{\star} \}$    satisfies  $\text{rank}(\widetilde{ \bf W}_P^{\star})=\text{rank}(\widetilde{ \bf P}_{F}^{\star} )=1$ and $\tilde{R}_{S,Ro}^{\star}={R}_{S,Ro}^{\star}$.  Meanwhile, the corresponding  original variables $\{{\tau}^{\star},{ \bf W}_P^{\star},
  { \bf P}_{F}^{\star} \}$   are    optimal for  problem \eqref{WSRM}.
\end{proposition}
  \vspace{-5mm}
  \begin{proof}
    Please refer to Appendix~\ref{AppenC}.
  \end{proof}  \vspace{-3mm}
Generally,  the  proof of Theorem~\ref{thre4} subject  to the complicated LMI constraint $\widetilde{\text{CR5}}$ is    more  difficult  than that of Theorem~\ref{thre3}.   Based on Theorem~\ref{thre4},   the  globally optimal solution $\{\!{\tau}^{\star},{ \bf W}_P^{\star},
  { \bf P}_{F}^{\star} \!\}$ to the  worst-case SRM problem \eqref{WSRM} can  also be  obtained 
   by successively solving the relaxed problem \eqref{WSRM5}, for which the proposed  convexity-based linear search in Section III  still works.
  \vspace{-4mm}
\subsection{The Proposed Outage-Constrained SRM }
It is widely recognized that  the worst-case   optimization  is the most conservative robust design, which  is  only encountered in practical   systems  with  a  low probability. Hence,   in this subsection,   we  consider the more general case of  statistically  imperfect  CSI, in which  the  FBS-EVE CSI error   is  assumed to be complex  Gaussian distributed, i.e.,  $\Delta{{ \bf   h}}_E=\text{vec}(\Delta{{ \bf   H}}_E^H)\sim\mathcal{CN}({\bf 0},{{ \bf   C}}_E)$,  where ${{ \bf   C}}_E\in
\mathbb{C}^{N_FN_E\times N_FN_E}$ denotes  the  positive semi-definite  error  covariance matrix.    Inspired by the fact  that  under 
  the   unbounded Gaussian  error  $\Delta{{ \bf   h}}_E$, an absolutely safe  beamforming design cannot be  guaranteed,   we instead consider the outage-constrained  SRM  to implement secure communications in  the  wirelessly powered  HetNet.  More specifically,  by defining the maximum secrecy rate outage probability  $p_{out}
 $, a $100(1\!-\!p_{out})\%$-safe design of our wirelessly powered HetNet  is  investigated.  Mathematically, the   outage-constrained SRM  problem is   formulated as
\begin{align}\label{OUTRM}
  &\underset{\substack{{\tau}_0,{ \bf   W}_P\succeq{ \bf   0}, { \bf   P}_{F}
  \succeq{ \bf   0} }} {\text{max}}~~~R_S,~~~~{\rm{s.t.}}~\text{CR1}\sim\text{CR3},~~\text{CR6}:~\text{Pr}_{{\Delta}{{ \bf   H}_E}}\{R_I-R_E\ge R_{S}\}\ge 1-p_{out}.
\end{align} 

Clearly, the secrecy outage constraint  $\text{CR6}$ indicates that  the
probability of the achievable secrecy rate being  over $R_S$
 should be higher than $1-p_{out}$,  and     problem \eqref{OUTRM} aims for maximizing this  $100p_{out}\%$-outage secrecy rate  threshold $R_S$.  However,   problem \eqref{OUTRM} is  computationally intractable   since the constraint $\text{CR6}$ does not have an explicit expression.  Therefore,  we consider replacing the function $R_I-R_E$ in  $\text{CR6}$  by an easy-to-handle function  via the following  Lemmas.  
\begin{lemma}\label{lemm3}\cite{Christensen:2008jwba}
For an arbitrary positive-definite matrix ${\bf E}\in\mathbb{C}^{N\times N}$,
we have
  $-\ln\det( {\bf E}) = \underset{{\bf S}\succeq{\bf 0}}{\text{max}}~
   -\text{tr}({\bf S}{\bf E})
+\ln\det({\bf S})+N$,
 where  the optimal ${\bf S}^{\star}$ is derived as  ${\bf S}^{\star}={\bf E}^{-1}$.
\end{lemma}\vspace{-5mm}
\begin{lemma}\label{lemm4}
{(Bernstein-type inequality (BTI)  \cite{Yuan:2018vzba})} For an arbitrary vector ${ \bf   x}\in \mathcal{CN}({\bf 0}, {\bf I})$,
    we assume $f({ \bf   x})={ \bf   x}^H{ \bf   A}{ \bf   x}+2\text{Re}\{{ \bf   x}^H{ \bf   b}\}+c$, where  ${ \bf   A}\in\mathbb{H}^{N\times N}$, ${ \bf   b}\in \mathbb{C}^{N}$ and $c$ is a constant. Then  for any $p_{out}\in[0,1]$,  the following convex approximation  holds, i.e. \vspace{-1mm}
\begin{align}\label{eqpower}
\small
\text{Pr}_{{ \bf   x}}\{f({ \bf   x})\!\le\! 0\}\!\ge\! 1\!-\!p_{out}\Longleftarrow\left\{\!\begin{array}{c}
\text{tr}({ \bf   A})+\sqrt{-2\ln p_{out}}  t_1-t_2\ln p_{out}+ c \le 0 \\
\left\Vert\left[ \begin{array}{l}\text{vec}({ \bf   A}) \\ \sqrt{2}{ \bf   u} \end{array}\right] \right\Vert_2 \le t_1\\
t_2{ \bf   I}_n-{ \bf   A}\succeq { \bf   0},~ {t}_2 \ge 0
\end{array}\right.,
\end{align}
 where $t_1$ and $t_2$ are a pair  of   slack variables.
\end{lemma}  \vspace{-3mm}
Firstly, by invoking Lemma 3, the  wiretap rate $R_E$  of  the EVE  can be   rewritten
 as \vspace{-1mm}
 \begin{align}\label{OUTRM2}
-R_E
  &=\frac{(1-\tau)}{\ln 2}\underset{{\bf S}\succeq \bm{0}}{\max} -\text{tr}\big[{\bf S}({{ \bf   I}}_{N_E}\!+\! {\bf R}_E^{\frac{1}{2}} { \bf   H}_{E}
  { \bf   P}_{F}{ \bf   H}_{E}^H{\bf R}_E^{\frac{1}{2}})\big]+
  \ln\det({\bf S})+N_E.
 \end{align}  We further
 substitute \eqref{OUTRM2} into   problem \eqref{OUTRM} to  obtain the reformulated constraint $\text{CR6}$ as
  \begin{align}\label{outsec1}
      \text{Pr}_{{    \Delta}{{\bf H}_E}}\left\{
  \text{tr}({ \bf   S} { \bf   R}_E^{\frac{1}{2}}{ \bf   H}_{E}
   { \bf   P}_{F}{ \bf   H}_{E}^H{ \bf   R}_E^{\frac{1}{2}} )-L\le 0\right\}\ge 1-p_{out}
 \end{align} where $L=\ln \big(1+ C_1{ { \bf   h}_{R}^H
   { \bf   P}_{F}
   { \bf   h}_{R}}\big)-\text{tr}({ \bf   S})+
\ln \det{{ \bf   S}}+N_E-\frac{\ln 2 R_S}{1-\tau}.$
 Since $\Delta{{ \bf   h}}_E\sim\mathcal{CN}({\bf 0}, {{ \bf   C}}_E)$, we can  re-express   $\Delta{{ \bf   h}}_E$ as  $\Delta{{ \bf   h}}_E\!=\!{{ \bf   C}}_E^{\frac{1}{2}}{ \bf   x}_e$ with ${ \bf   x}_e\!\sim \!\mathcal{CN}({\bf 0}, {\bf I}_{N_FN_E})$. Furthermore, through the vectorization of \eqref{outsec1}, we  have
 \vspace{-1mm}
\begin{align}\label{outc2}
   \text{Pr}_{{    \Delta}{{\bf x}_e}}\left\{
    \Delta{ \bf   x}_e^H{{ \bf   C}}_E^{\frac{1}{2}} { \bf   P}_S {{ \bf   C}}_E^{\frac{1}{2}}\Delta{ \bf   x}_e +2\text{Re}\{
    \Delta{ \bf   x}_e^H{{ \bf   C}}_E^{\frac{1}{2}}{ \bf   P}_S\hat{{ \bf   h}}_E \}+
    \hat{{ \bf   h}}_E^H { \bf   P}_S \hat{{ \bf   h}}_E-L \le
     0\right\}\ge 1-p_{out}, 
 \end{align} where $
{ \bf   P}_S= ({ \bf   R}_E^{\frac{1}{2}}{ \bf   S}{ \bf   R}_E^{\frac{1}{2}})^T\otimes { \bf   P}_F$. To
make  the probabilistic constraint \eqref{outc2}  tractable, we  adopt
a popular conservative  approximation, the so-called  BTI in Lemma~\ref{lemm4},  for   transforming  it into
a series of  tractable convex  constraints.  Then the outage-constrained SRM problem \eqref{OUTRM}   is      reformulated as \vspace{-5mm}
\begin{align}\label{OUTSRM2} 
&\underset{\substack{t_1,t_2,{\tau},{ \bf   W}_P\succeq{ \bf   0}, \\ { \bf   P}_{F}
  \succeq{ \bf   0}, {\bf S}\succeq { \bf   0} }} {\text{max}}~R_S, ~
{\rm{s.t.}}~\text{CR1}\!\sim\!\text{CR3},~\text{CR6}:{\small\left\{\!\begin{array}{l}
\text{tr}( \widehat{{ \bf   H}}_{CE}{ \bf   P}_S  )\!+\!\sqrt{-2\ln p_{out}} t_1\!-\! t_2\ln p_{out} -\ln (1\!+\!C_1{ { \bf   h}_{R}^H
   { \bf   P}_{F}
   { \bf   h}_{R}}) \\
    +\text{tr}({ \bf   S})\!-\!\ln\det{ \bf   S}\!-\!N_E\!+\!\frac{\ln 2 R_S}{1-\tau} \le 0 \\
\left\Vert\left[ \begin{array}{l}\text{vec}( {{ \bf   C}}_E^{\frac{1}{2}}{ \bf   P}_S{{ \bf   C}}_E^{\frac{1}{2}} ) \\ \sqrt{2}{{ \bf   C}}_E^{\frac{1}{2}}{ \bf   P}_S \hat{{ \bf   h}}_E \end{array}\right] \right\Vert_2 \le t_1\\
t_2{ \bf   I}_{N_FN_E}-{{ \bf   C}}_E^{\frac{1}{2}}{ \bf   P}_S{{ \bf   C}}_E^{\frac{1}{2}} \succeq { \bf   0}, ~t_2 \ge 0
\end{array}\right.,}
\end{align} where $\widehat{{ \bf   H}}_{CE}\!=\!\hat{{ \bf   h}}_E\hat{{ \bf   h}}_E^H+{{ \bf   C}}_E$. Although     problem \eqref{OUTSRM2}  is still not  jointly  concave w.r.t $\{{\tau},{ \bf   W}_P, { \bf   P}_{F},$ $ {\bf S}\}$, it is more  tractable  than the  outage-constrained SRM problem \eqref{OUTRM}.  
Specifically, by  utilizing     $\widetilde{\bf W}_P\!=\!
\tau{\bf W}_P$, $\widetilde{\bf P }_{F}
\!=\!(1\!-\!\tau) {\bf P}_{F}$ and $\widetilde{\bf S}\!=\!
\tau{\bf S}$,  we easily find that   problem \eqref{OUTSRM2}  is jointly concave w.r.t. $\{\tau,\widetilde{\bf W}_P,\widetilde{\bf P }_{F}\}$  when fixing ${\bf S}$, which is  similar to  problems \eqref{SRMRe3} and \eqref{WSRM5}.  In turn,  problem \eqref{OUTSRM2}  is also  concave  w.r.t. ${\bf S}$ when fixing $\{\tau, {\bf W}_P,{\bf P }_{F}\}$. Interestingly, we also  validate that  problem  \eqref{OUTSRM2} admits the optimal   rank-1  solution, as shown in the following Theorem.  \vspace{-3mm}
\begin{proposition}\label{thre5}
The optimal solution $\{{ \bf W}_P^{\star}, { \bf P}_{F}^{\star} \}$ to  problem \eqref{OUTSRM2}  satisfies $\text{rank}({ \bf W}_P^{\star})\!=\!\text{rank}( { \bf P}_{F}^{\star} )\!=\!1$,  which is also a high-quality solution  for the computationally  intractable problem  \eqref{OUTRM}. 
\end{proposition}
  \vspace{-3mm}
  \begin{proof}
    Please refer to Appendix~\ref{AppenC}.
  \end{proof} 

Based on the above analysis, we  propose an alternating  optimization procedure  for efficiently solving  problem \eqref{OUTSRM2} in the sequel. The  first subproblem   for   optimizing  $\{\tilde{t}_1,\tilde{t}_2,\tau,\widetilde{\bf W}_P,\widetilde{\bf P }_{F}\}$  given  ${\bf S}$ is  expressed as
\begin{align}\label{OUTSRM3}
&\underset{\substack{\tilde{t}_1,\tilde{t}_2,{\tau}, \widetilde{ \bf   W}_P\succeq{ \bf   0}, \widetilde{ \bf   P}_{F}
  \succeq{ \bf   0} }} {\text{max}}~~~~ R_S,\nonumber\\&~~~~{\rm{s.t.}}~~~\widetilde{\text{CR1}}\!\sim\!
\widetilde{\text{CR3}},~\widetilde{\text{CR6}}:\!\!{\small\left\{\begin{array}{l}
\widetilde{\text{CR6}}a:\text{tr}( \widehat{{ \bf   H}}_{CE}\widetilde{ \bf   P}_S  )+\sqrt{-2\ln p_{out}} \tilde{t}_1- \tilde{t}_2\ln p_{out}
    \!+\!(1\!-\!\tau)\text{tr}({ \bf   S})\\
  \!-\!(1\!-\!\tau)\ln\det{ \bf   S}\!-\!(1\!-\!\tau)N_E\!+\!{\ln 2 R_S} \!\le\! (1\!-\!\tau)\ln (1\!+\!\frac{ C_1{ \bf   h}_{R}^H
  \widetilde { \bf   P}_{F}
   { \bf   h}_{R}}{1\!-\!\tau}) \\
\widetilde{\text{CR6}}b: \left\Vert\left[ \begin{array}{l}\text{vec}( {{ \bf   C}}_E^{\frac{1}{2}}\widetilde{ \bf   P}_S{{ \bf   C}}_E^{\frac{1}{2}} ) \\ \sqrt{2}{{ \bf   C}}_E^{\frac{1}{2}}\widetilde{ \bf   P}_S \hat{{ \bf   h}}_E \end{array}\right] \right\Vert_2 \le \tilde{t}_1\\
\widetilde{\text{CR6}}c: \tilde{t}_2{ \bf   I}_{N_FN_E}\!-\!{{ \bf   C}}_E^{\frac{1}{2}}\widetilde{ \bf   P}_S{{ \bf   C}}_E^{\frac{1}{2}} \succeq { \bf   0}, ~\tilde{t}_2 \ge 0,
\end{array}\right.},
\end{align}
where $\widetilde { \bf   P}_{S}\!=\!({ \bf   R}_E^{\frac{1}{2}}{ \bf   S}{ \bf   R}_E^{\frac{1}{2}})^T\otimes \widetilde{ \bf   P}_F$. Due to the joint   concavity of  problem \eqref{OUTSRM3}, we   can  optimally recover  a  high-quality suboptimal solution $\{\tau, { \bf   W}_{P}, { \bf   P}_{F}\}$ to  problem \eqref{OUTRM} from    the  obtained    optimal rank-1 solution  of problem \eqref{OUTSRM3}.   Then the second  subproblem  for optimizing ${ \bf   S}$ given $\{\tau, { \bf   W}_{P}, { \bf   P}_{F}\}$ is   formulated as  
\begin{align}\label{OUTSRM4}
&\underset{\substack{\tilde{t}_1,\tilde{t}_2,
 {\bf S}\succeq{ \bf   0} }} {\text{max}}~~~~R_S,~~~~~{\rm{s.t.}}~~\text{CR1}\sim\text{CR3},~\widetilde{\text{CR6}}.
\end{align}

Since    problem   \eqref{OUTSRM4}  is also  jointly  concave w.r.t. $\{\tilde{t}_1,\tilde{t}_2,
 {\bf S}\succeq{ \bf   0} \}$, it  is concluded that the global optimality at each iteration  is guaranteed  and  the  achievable  secrecy rate $R_S$ is nondecreasing  within the whole iterative process. Moreover, considering  the closed and finite feasible   region of  problem \eqref{OUTSRM2}, the proposed alternating optimization procedure for problem \eqref{OUTSRM2}
is guaranteed to  converge to a locally optimal secrecy rate value.
\vspace{-4mm}
\section{ Secure Communications under    Completely  Unknown   FBS-EVE  CSI}\label{S5}
  In the previous two Sections,  both full  and  partial FBS-EVE CSI are considered  for jointly optimizing  the PB energy covariance matrix, the FBS information covariance matrix and the  time splitting factor.  Nevertheless, in  practice, the EVE may be passive and the FBS is unaware of the existence of  the EVE.  In this context, we  accordingly  propose   artificial noise (AN) aided scheme  for secrecy wirelessly powered HetNets. With  this scheme, the FBS simultaneously transmits  the confidential  signal  ${ \bf s}_F$ and  AN ${ \bf z}_{F}$ to the FU, i.e., ${ \bf x}_F\!=\!{ \bf s}_F\!+\!{ \bf z}_F$,   for   sufficiently  interfering the EVE   without affecting the FU. Hence,    the generated AN should be in the null-space of the FBS-FU channel, i.e., ${\bf h}_R {\bf z}_F\!=\!{\bf 0}$.  We  then  express the AN   as  ${\bf z}_F\!=\!{\bf H}_R^{\perp}{\bf n}_F$, where ${\bf H}_R^{\perp}\!\in \! \mathbb{C}^{N_F\!\times\! (N_F\!-\!1)}$ denotes the orthogonal complement space of   ${\bf h}_R$ and ${\bf n}_F\!\sim\! \mathcal{CN}({\bf 0}, { \bf   \Sigma}_{U})$ is a  novel $(N_F\!-\!1)$-dimensional  AN vector  with  the positive semi-definite  covariance matrix ${ \bf   \Sigma}_{U}\!\succeq \! {\bf 0}$. Since the FBS-EVE CSI is unavailable at the FBS, the direct secrecy rate optimization  of  the   considered wirelessly powered HetNet is infeasible.
  As an alternative, we consider   maximizing the  average  AN power  $(1\!-\!\tau)\text{tr}( {{ \bf   \Sigma}}_{U})$  at the EVE subject to the minimum legitimate rate at the FU,   so as to  impose  as much  interference as possible    on   the potential
EVE  for  reducing   its   wiretaping  capability.  The AN aided secrecy problem  is  ultimately formulated  as 
\begin{align}\label{ARSRM2}
  &\underset{\substack{{\tau},{ \bf   W}_P\succeq{ \bf   0},
  { \bf   P}_{F}\succeq{ \bf   0},{ \bf   \Sigma}_U\succeq{ \bf   0} }}
{\text{max}}~~~(1\!-\!\tau)\text{tr}( {{ \bf  \Sigma}}_{U})\nonumber\\
&~~~~~~~~~{\rm{s.t.}}~~{\text{CR1}}, {\text{CR3}},~{\text{AR2}}:
(1-\tau)\log_2\bigg(1+ \frac{ { \bf   h}_{R}^H
   { \bf   P}_{F}
   { \bf   h}_{R}}{\sigma_n^{2}+\frac{P_M}{N_M}\Vert {\bf g }_R\Vert^2}\bigg)
\ge {R}_{th}\nonumber\\
 &~~~~~~~~~~~~~~~{\text{AR3}}:(1\!-\!\tau)\text{tr}({{ \bf   P}}_{F}\!+\! { \bf   \Sigma}_{U})\!
 \le \tau \xi \big[\text{tr}({ \bf    H}_F{ \bf    W}_P{ \bf    H}_F^H)+\frac{P_M}{N_M}\text{tr}(
  { \bf    G}_F{ \bf    G}_F^H)\big]
\end{align} 

Naturally,  problem \eqref{ARSRM2} with   coupled variables is not jointly concave w.r.t $\{{\tau},{ \bf   W}_P,
  { \bf   P}_{F},{ \bf   \Sigma}_U \}$. Based on the definitions $ \widetilde{ \bf   W}_P=\tau{ \bf   W}_P $, $  \widetilde{ \bf   P}_{F}=(1-\tau)  { \bf   P}_{F}$  and $\widetilde{ \bf \Sigma}_{U}=(1-\tau)  { \bf  \Sigma}_{U}$,   problem (\ref{ARSRM2}) can  be   rewritten  as
  \vspace{-2mm}
\begin{align}\label{ARSRM3}
  &\underset{\substack{{\tau},\widetilde{ \bf   W}_P\succeq{ \bf   0},
  \widetilde{ \bf   P}_{F}\succeq{ \bf   0},\widetilde{ \bf   \Sigma}_U\succeq{ \bf   0} }}
{\text{max}}~~~\text{tr}( \widetilde{{ \bf   \Sigma}}_{U})\nonumber\\
&~~~~~~~~~{\rm{s.t.}}~~\widetilde{\text{CR1}}, \widetilde{\text{CR3}},~\widetilde{\text{AR2}}:
  (1-\tau)\log_2\bigg(1+ \frac{ C_1{ \bf   h}_{R}^H
  \widetilde { \bf   P}_{F}
   { \bf   h}_{R}}{1-\tau}\bigg) \ge {R}_{th}\nonumber\\
 &~~~~~~~~~~~~~~~\widetilde{\text{AR3}}:~\text{tr}(\widetilde{{ \bf   P}}_{F}\!+\! \widetilde{ \bf   \Sigma}_{U})\!
 \le \xi \text{tr}({ \bf    H}_F\widetilde{ \bf    W}_P{ \bf    H}_F^H)+\xi \tau\frac{P_M}{N_M}\text{tr}(
  { \bf    G}_F{ \bf    G}_F^H).
\end{align}

Similar to   problems \eqref{SRMRe3} and \eqref{WSRM5},  problem \eqref{ARSRM3}  composed of  a  linear objective function and  concave  constraints  is also jointly concave w.r.t $\{{\tau},\widetilde{ \bf   W}_P,
  \widetilde{ \bf   P}_{F},\widetilde{ \bf   \Sigma}_U \}$, and thus can be optimally solved. Moreover, the optimal  rank-1  solution  $\{\widetilde{ \bf   W}_P^{\star}, \widetilde{ \bf   P}_{F}^{\star},\widetilde{ \bf   \Sigma}_U^{\star} \}$ to   problem  \eqref{ARSRM3}  is revealed   in the following  Theorem.
  \vspace{-3mm}
 \begin{proposition}\label{thre5}
 There always  exists an optimal solution $\{{ \bf   W}_P^{\star}, { \bf   P}_{F}^{\star},{ \bf   \Sigma}_U^{\star} \}$ with $\text{rank}({ \bf   W}_P^{\star})=\text{rank}({ \bf   P}_{F}^{\star})=\text{rank}({ \bf   \Sigma}_U^{\star})=1$ for   the  AN aided secrecy  problem   \eqref{ARSRM2}.
\end{proposition}
  \vspace{-3mm}
  \begin{proof}
    Please refer to Appendix~\ref{AppenD}.
  \end{proof}
  \vspace{-3mm}
  
 \begin{proposition}\label{total}
For all the above cases of FBS-EVE CSI, if  the  cross-interference constraint ${\text{CR3}}$  is inactive,   a closed-form      expression for the PB energy covariance  matrix  can be derived as  ${ \bf   W}_P^{\star}=P_P{{\bf w }_P{\bf w }_P^{ H }}$ with ${\bf w }_P=\nu_{\max}({\bf H}_F^H{\bf H}_F)$.
\end{proposition}
  \begin{proof}
Please refer to Appendix~\ref{AppenE}.
  \end{proof}

\section{Simulation Results and Discussions} \label{simulation}
In this section, numerical  simulation results are provided for   quantifying the     secrecy rate performance  of the   wirelessly powered HetNet.
Unless otherwise stated, in the following simulations the MBS, the PB, the FBS and the EVE   are equipped  with $N_M=2$, $N_P=4$, $N_F=4$  and $N_E=3$ antennas, respectively.  Furthermore, a  single-antenna MU and a single-antenna  FU  are considered.
The locations of the  MBS,  the MU and the PB  are   $(0, 0)$m, $(150, 0)$m and $(100, 100)$m, respectively.  Since the FBS  mainly  harvests  wireless energy from the PB, the FBS  location is near the PB  and thus  set as $(105,105 )$m. Moreover, the FU and the EVE are located  at $(105, 200)$m and $(155, 105)$m, respectively.
For the perfect SRM (PSRM), we assume that  all   channel coefficients  follow the i.i.d.  Rayleigh distribution with zero mean and  $ 10^{-3}d^{-\alpha}$ variance, where $d$ denotes the actual  distance between two arbitrary terminals  and $\alpha=2$ is the pathloss exponent.  The  Gaussian noise variance is  set as $\sigma_n^2=-110$dBm.  The maximum transmit   power of the MBS and  PB  are  defined as $P_M=10$dBm and $P_P=20$dBm, respectively. The FBS energy harvesting efficiency is $\xi=0.8$.  Initially,  we fix the interference threshold  to be  $I_{th}= -65$dBm.  For the  worst-case SRM (WSRM), the bound of    FBS-EVE CSI  error is set as $\xi_f= 0.2$. By contrast, for the outage-constrained SRM (OSRM),    the covariance matrix  of   FBS-EVE CSI error  ${{ \bf   C}}_E$ and  the outage probability $p_{out}$ are defined  as   ${{ \bf   C}}_E=\zeta^2{\bf I}_{N_E}$ with $\zeta=0.05$ and $p_{out}=0.1$, respectively.  Finally,  for the proposed AN scheme,  an    legitimate  rate threshold at the FU is  chosen as $R_{th}=\frac{R_{\max}}{2}$, where $R_{\max}$ denotes the maximum  achievable rate  of the  wirelessly powered HetNet without the   EVE  and    can be optimally derived  according to \cite{Lohani:2016etba}. Note that all simulation   points  are plotted  by averaging over  200 Monte-Carlo experiments.

Additionally, we   consider   four   benchmarks for  the proposed secrecy beamforming designs    under  three different  levels of FBS-EVE CSI. Specifically, for both   the  perfect and unknown  FBS-EVE CSI cases,   the  equal  power allocation  schemes  are adopted at the FBS ( also referred to as FBS-EPA)  with the fixed   and optimized time allocation, i.e. $\tau=\frac{1}{2}$  and    $\tau=\tau^{\star}$, respectively. Moreover, both the two  benchmarks  assume  ${\bf W}_P=\frac{P_P}{N_P}{\bf I}_{N_P}$.  While   for  the imperfect FBS-EVE CSI,   the nonrobust SRM scheme in \cite{Lohani:2016etba}  is  adopted as  a benchmark   for the  proposed WSRM, where the achievable secrecy rate is  calculated  by  substituting   the   optimal  solution  $\{{\tau},{ \bf   W}_P,
  { \bf   P}_{F}\}$   derived when  $\xi=0$  into the worst-case FBS-EVE channel.  Furthermore, by referring to \cite{Tabassum:2015tsba}, the  worst-case  SRM with the 
  error bound  $\xi=\sqrt{\frac{\eta}{2}F_{\chi_{2N_FN_E}^{2}}^{-1}(1-p_{out})}$ is  actually  a  conservative  approximation for the proposed OSRM, where   $F_{\chi_{2N_FN_E}^{2}}^{-1}(\cdot)$ denotes  the inverse cumulative density function (CDF) of a Chisquare random variable with degrees of freedom $2N_FN_E$.   Therefore,   we  also consider  it   to be a benchmark  for  the proposed OSRM. \footnote {Notice that 
if the predefined interference level $I_{th}$   is not realized  by  the adopted benchmarks,  the corresponding achievable secrecy rate is set  to  zero.}
  \vspace{-4mm}
  \begin{figure}[tbp]
    \vspace{-8mm}
     \centering                                               \includegraphics[width=3.5in]{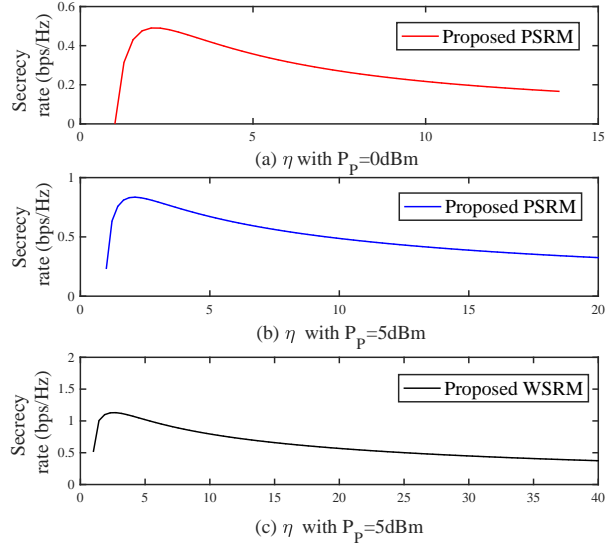}
    \vspace{-6mm}
     \caption{Achievable secrecy rates of the proposed  PSRM and WSRM  versus  the parameter $\eta$ to validate Theorem~\ref{thre2}.} %
     \label{Fig2}            \vspace{-8mm}
     \end{figure}
\subsection{The Proposed Perfect and Worst-Case  SRMs}

In this subsection, we evaluate the achievable secrecy rates of the proposed PSRM and WSRM  in Section~\ref{III} and Section~\ref{IV_A} corresponding to  problem \eqref{SRMRe3} and problem \eqref{WSRM5},  respectively.
Firstly, to validate  Theorem~\ref{thre2},  the achievable  secrecy rates of the above two schemes 
over   a feasible range of $\eta$ are  shown  in Fig.~\ref{Fig2}, where different PB transmit power  $P_P=0,5$dBm  are considered.
 It  is clearly  observed from Fig.~\ref{Fig2} that
both the  PSRM problem \eqref{SRMRe3} and  WSRM problem \eqref{WSRM5} are  unimodal (quasi-concave)  functions  of $\eta$, and thus  the   globally  optimal $\eta$  can be  determined   via the   GSS. 
\begin{figure}[tbp]
\vspace{-6mm}
\centering                                                       \subfigure[$N_E=3$]{
 \begin{minipage}{8cm}
\centering                                      
\includegraphics[scale=0.45]{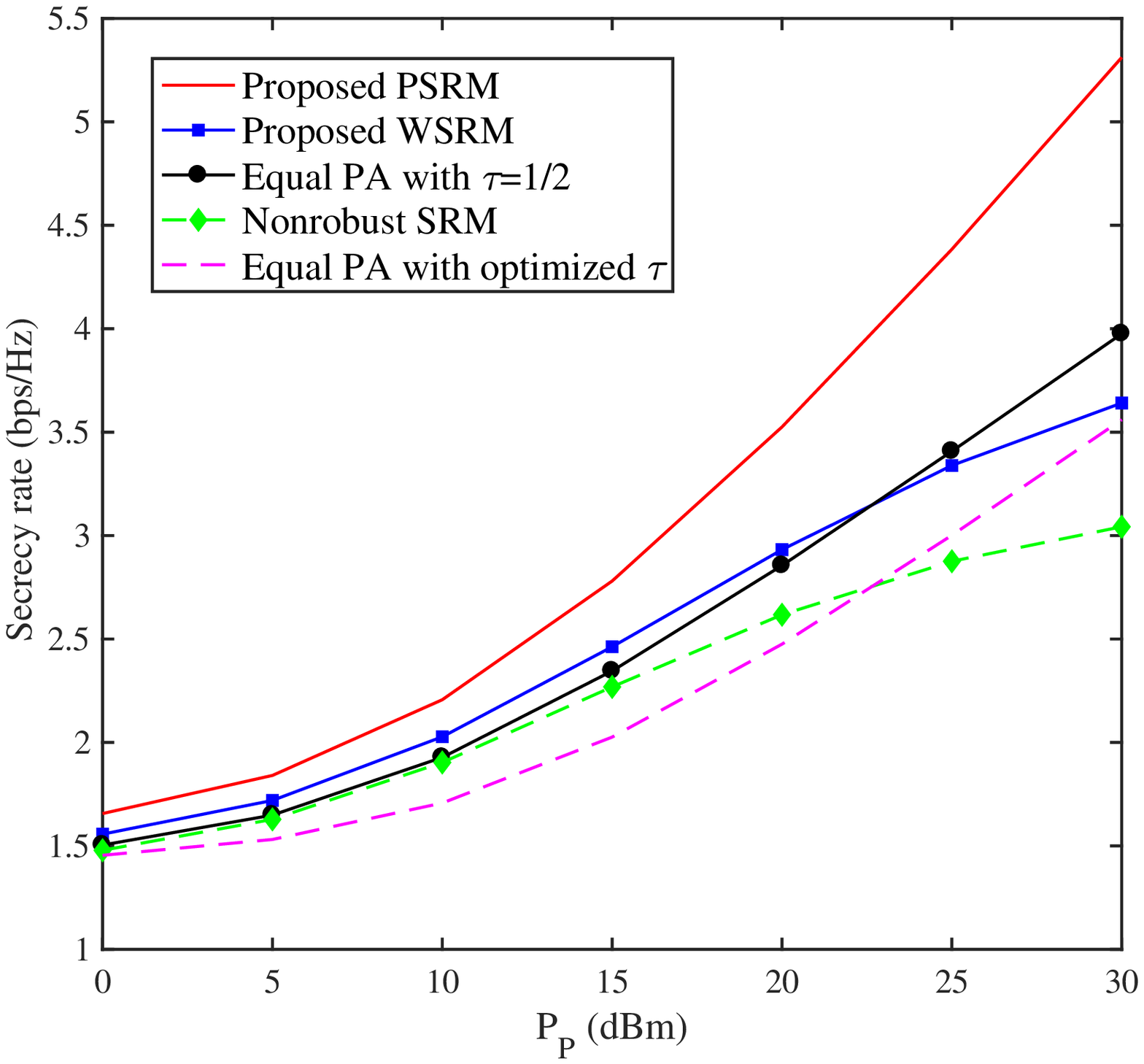} \end{minipage}}
\subfigure[\vspace{5mm} $N_E=4$]{
\begin{minipage}{8cm}
\centering                      
\includegraphics[scale=0.45]{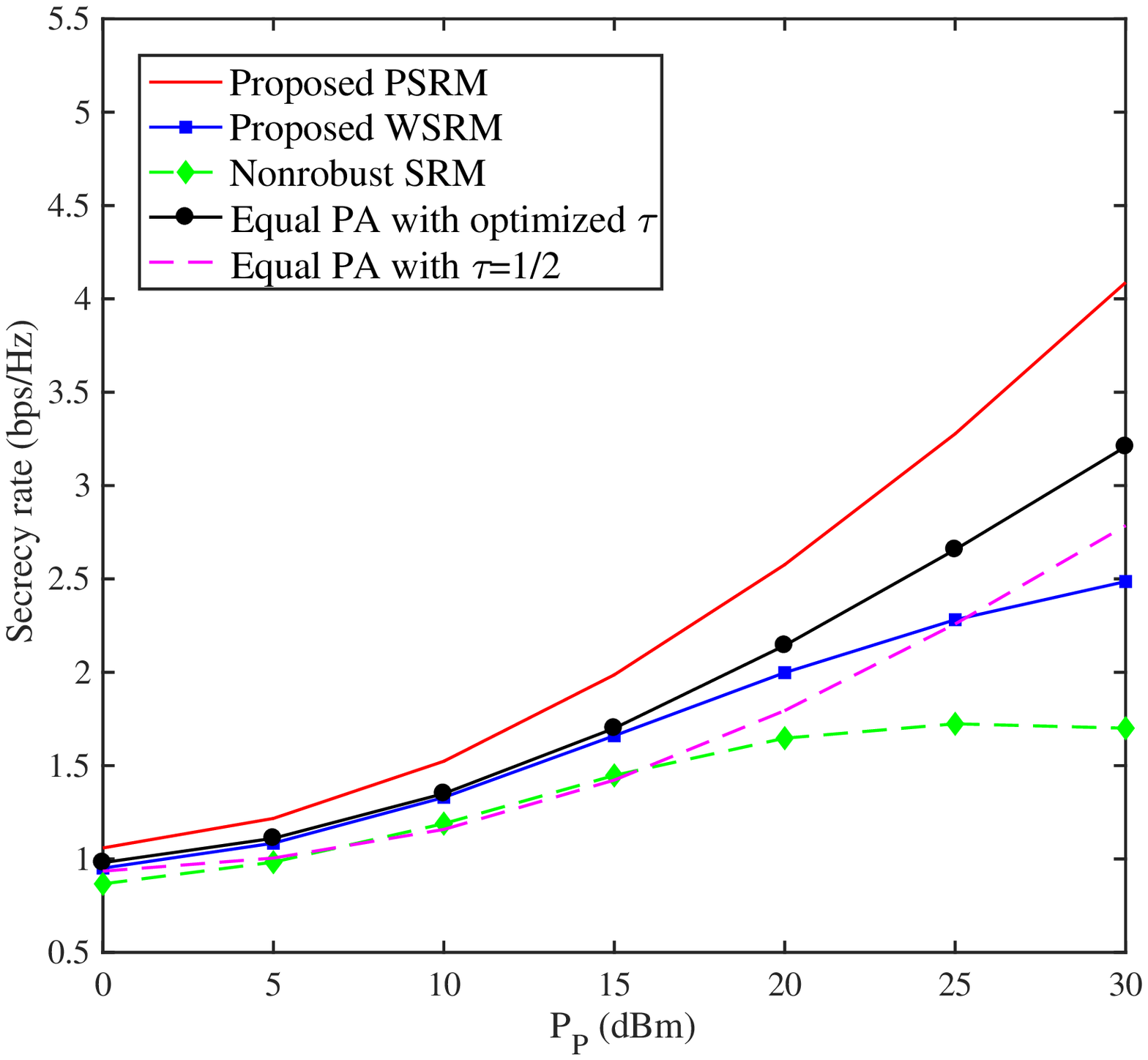} \end{minipage}}
\vspace{-6mm}
\caption{Achievable  secrecy rates for  different  schemes  versus  the PB transmit power $P_P$ with (a) $N_E=3$ and (b)  $N_E=4$.} %
\label{Fig3}             
\vspace{-8mm}
\end{figure}

Next, Fig.~\ref{Fig3} depicts   the  achievable  secrecy rates  of  all   schemes versus the PB
 transmit power $P_P$ with  different numbers of EVE's antennas $N_E$.  From Fig.~\ref{Fig3}\,(a) with $N_E=3$, we   readily find that the achievable secrecy rates   of all schemes  increase with  $P_P$. 
   The proposed PSRM    realizes  the highest secrecy  rate 
among all schemes.  While  the proposed WSRM  considering the influence  of  practical CSI error on the  secrecy beamforming design   naturally  performs better than the nonrobust counterpart.  By  comparing the two adopted  FBS-EPA benchmarks, we clearly see  that    time allocation  plays an   important role in    improving   secrecy rate performance  of the wirelessly powered HetNet.  Furthermore,  in Fig.~\ref{Fig3}\,(b) with $N_E=4$, the same comparisons  among all schemes as in Fig.~\ref{Fig3}\,(a) can  be observed. Moreover,   each scheme achieves  a lower  secrecy rate  than  its   counterpart in Fig.~\ref{Fig3}\,(a)  due  to  the increased  EVE's wiretap capability.
\begin{figure}[tp!]\vspace{-10mm}
\centering
\includegraphics[width=3.5in]{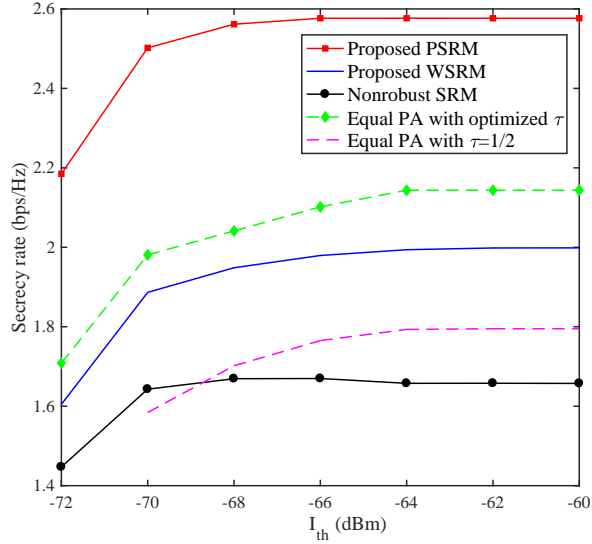}
\vspace{-7mm}
\caption{Achievable secrecy rates for  different  schemes  versus the interference level  $I_{th}$. }\label{Fig5}
\vspace{-3mm}
\end{figure}
\begin{figure}[tp!]
\centering
\includegraphics[width=3.5in]{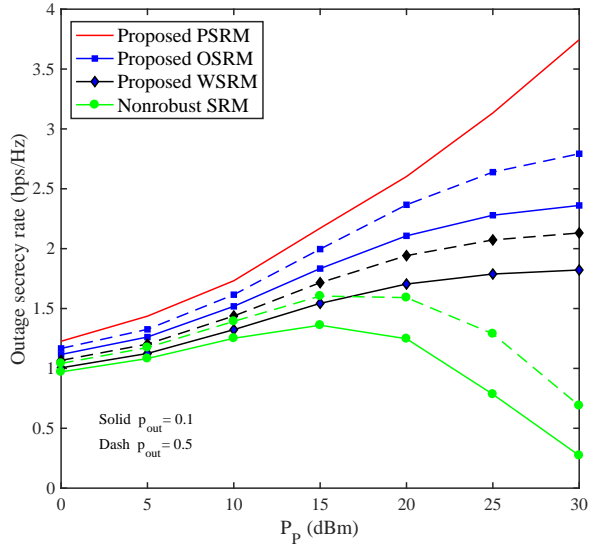}
\vspace{-8mm}
\caption{Achievable secrecy rates for  different schemes  versus the PB transmit power $P_P$ under different outage probabilities. } \label{Fig6}
\vspace{-8mm}
\end{figure}

Fig.~\ref{Fig5}  investigates the achievable  secrecy rates    of all  schemes  versus   the   interference level $I_{th}$. We  clearly find from Fig.~\ref{Fig5} that the  secrecy  rate value for each scheme   firstly increases with $I_{th}$ until  $I_{th}$ reaches a certain threshold  $I_{th}^{mi}$,  then  becomes saturated. Here,  $I_{th}^{mi}$  represents   the optimal solution  of the    PSRM problem \eqref{SRMRe3} or the WSRM problem \eqref{WSRM5}  with  the interference constraint  neglected.    This is because that when a small $0<I_{th}\le I_{th}^{mi}$ is adopted, it is easily inferred that the achievable  secrecy  rates of all  schemes are   dominated  by the  cross-interference constraint. In this context, since the  feasible  regions of both problems  \eqref{SRMRe3} and \eqref{WSRM5}  expand    with the increasing  $I_{th}$, the corresponding achievable  secrecy  rates also increase. However, when $I_{th}$ becomes sufficiently large, i.e $I_{th}>I_{th}^{mi}$, the  cross-interference constraint  actually becomes   inactive and  thus has no effect on   achievable secrecy  rates.  That is to say,   the achievable secrecy  rates  become saturated  regardless of the increasing $I_{th}$ when $I_{th}>I_{th}^{mi}$.  In particular,  the saturated  secrecy rate value  can    be   derived  from    problem \eqref{SRMRe3}(\eqref{WSRM5}) by neglecting the  cross-interference constraint.  Additionally,   the proposed PSRM   still has   the  best  secrecy rate  performance  among
 all   schemes  and the proposed WSRM   provides   much  higher secrecy rate  than  the nonrobust SRM.
\vspace{-4mm}
\subsection{The Proposed Outage-Constrained SRM }
In this subsection, we investigate the  achievable secrecy rates  of the proposed OSRM, i.e. problem \eqref{OUTRM} in Section IV. B.  As  mentioned before, the proposed WSRM with FBS-EVE's  CSI error bound $\xi_f=\sqrt{\frac{\eta}{2}F_{\chi_{2N_FN_E}^{2}}^{-1}(1-p_{out})}$ can  be  regarded as another conservation  reformulation   for  the intractable outage probability constraint, so we  also adopt it  as a benchmark    for fair comparisons hereafter.

 Fig.~\ref{Fig6} shows secrecy rate performance  for  all schemes  versus PB transmit power $P_P$  under  different outage probabilities $p_{out}\!=\!0.1, 0.5$.  We firstly observe  from Fig.~\ref{Fig6}
 that  for  both  considered  $p_{out}$,  the  achievable secrecy rates of the proposed PSRM, WSRM and OSRM  all increase with $P_P$.  Meanwhile, the proposed OSRM    has  a better secrecy rate  performance  than the proposed WSRM, while  the nonrobust  SRM  performs worst due to the ignorance of  CSI error.  Furthermore,  we find that for each scheme the higher  secrecy rate  is realized   with  the 
   outage probability $P_{out}\!=\!0.5$. This is because that  when the outage tolerance is  relaxed, i.e. from $P_{out}\!=\!0.1$ to $P_{out}\!=\!0.5$,  the corresponding  secrecy outage rate threshold also becomes  higher.
\begin{figure}[t]                                                           
\vspace{-10mm}
\centering
\includegraphics[width=3.5in]{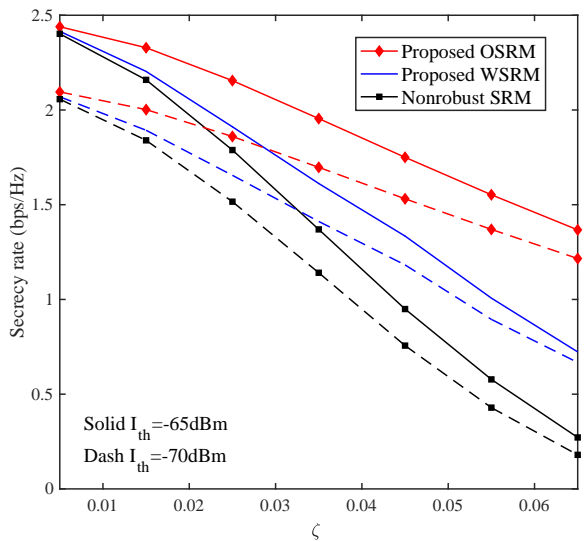}\vspace{-8mm}
\caption{ Achievable secrecy rates for  different  schemes  versus the  CSI error variance $\zeta$  under different interference levels $I_{th}$. } %
 \label{Fig7}
 \vspace{-8mm}
\end{figure}
\begin{figure}[t] \vspace{-8mm}
\centering
\includegraphics[width=3.5in]{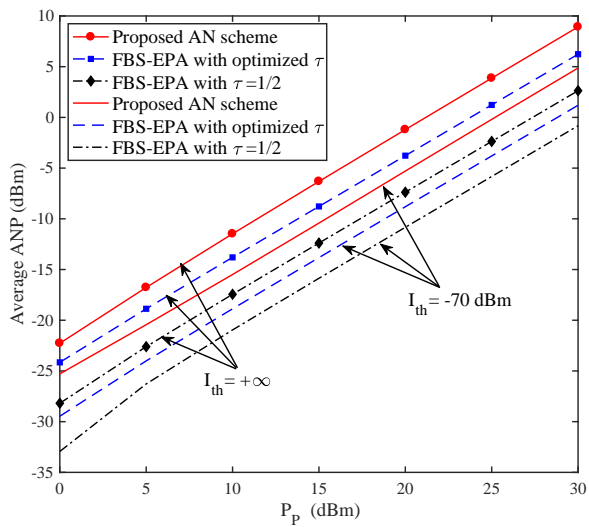}\vspace{-5mm}
\caption{ Artificial noise power for  different  schemes  versus the PB transmit power $P_P$   under different interference levels $I_{th}$, where   $R_{th}=1.5$(bps/Hz).} %
  \label{Fig8}
  \vspace{-8mm}
\end{figure}
Fig.~\ref{Fig7}  illustrates   secrecy rate performance  of  all  schemes  versus  the   CSI  error  variance  $\zeta$   with    different interference levels $I_{th}=-65,-70$dBm.  On the one hand, for each $I_{th}$, we observe  that the  achievable secrecy rates  of all   schemes  decrease  with the increasing $\zeta$.   The proposed OSRM  still  has better  secrecy rate performance   than the  proposed WSRM. While   the nonrobust SRM   performs the worst. Moreover, the performance   gaps  between the   proposed OSRM  and the  other two  schemes   both become  larger   when $\zeta$ increases.   On the other hand, the   higher  secrecy rates  of all  schemes  are observed at   the  interference level $I_{th}\!=\!-65$dBm  due to the larger feasible  region  of  our  secrecy beamforming design.
\vspace{-5mm}
\subsection{The proposed AN scheme }
\begin{figure}[t]
\centering
\includegraphics[width=3.5in]{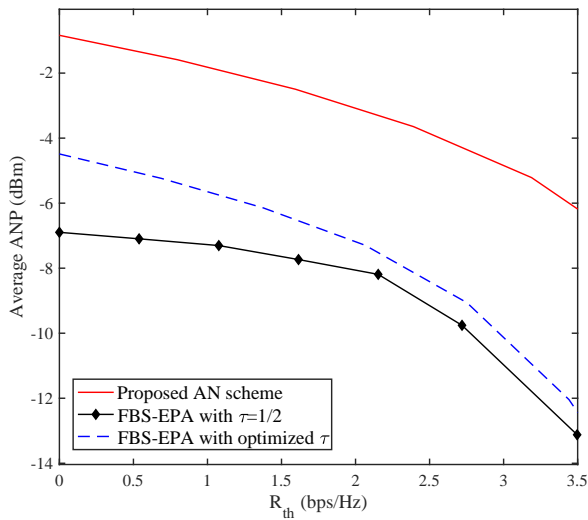}
\vspace{-8mm}
\caption{ Artificial noise power for   different schemes  versus the legitimate  rate threshold $R_{th}$.}
  \label{Fig10}
  \vspace{-8mm}
\end{figure}
 In this subsection, our goal is to  illustrate the  optimal  artificial noise  power (ANP)   of the  proposed AN scheme in Section~\ref{S5}.  In Fig.~\ref{Fig8},  the optimized   ANP  for  different schemes  versus the  PB transmit power $P_P$  with  different  interference thresholds  $I_{th}=-70, +\infty$ is  shown, where the legitimate  rate threshold is set as  $R_{th}=1.5$(bps/Hz). \footnote{ Notice that for  some small  $P_P$, by which  the legitimate rate threshold    $R_{th}$ is not supported, we  correspondingly   set  the achievable secrecy rate  to be zero.} We firstly find from Fig. ~\ref{Fig8} the achieved ANP increases with  $P_P$ for all three schemes. Moreover, the proposed AN scheme  has the highest ANP  compared to the other two schemes,   implying  that   the strongest  interference  is imposed on   the  EVE  to   reduce    its wiretapping capability.

Fig.~\ref{Fig10}   shows the  optimal  ANP  for   different schemes versus the legitimate rate threshold $R_{th}$. It is clear from Fig.~\ref{Fig10}  that  for each scheme,  the optimized  ANP   decreases with
the increasing   $R_{th}$  since  the more  portion of   PB transmit   power $P_P$ is allocated for information transfer  to  the  FU for  achieving  the  rate threshold   $R_{th}$. Furthermore, the proposed AN scheme still realizes the highest ANP  among all  schemes, which promotes  secure communications  of the wirelessly powered HetNet via  sufficiently interfering  the  malicious Eve.
\vspace{-4mm}
\section{Conclusions}
In this paper,  we   investigated  secure communications of the wirelessly powered HetNet from  different perspectives of the FBS-EVE  CSI. Firstly, in the case of    perfect  FBS-EVE CSI, the perfect  and   worst-case SRMs  was studied  by jointly optimizing the PB energy and FBS information covariance matrices as well as the time splitting factor, which  can be  globally addressed    by the proposed convexity-based linear search.  Secondly, considering    deterministically  and  statistically  imperfect cases of  FBS-EVE  CSI,   we applied S-procedure and  BTI   to transform  intractable CSI  error related  constraints into a series of tractable convex ones, respectively.      Finally, for   the completely  unknown  FBS-EVE  CSI case,  an  AN aided secrecy  beamforming design  was proposed  to interfere the  EVE as much as possible. In particular,   
when the cross-interference  constraint  is negligible, the closed-form   PB energy covariance matrix for all conisidred cases can be derived. More importantly, we   also proved    the rank-1 property of the optimal solutions for all studied SRM problems.  Numerical experiments   verified  the excellent  secrecy performance of  all our proposed  secrecy beamforming designs.
\begin{appendices}
\section{  }\label{AppenA}
Firstly,  assuming  the fixed $\tau$ and  $\eta$,  we can  reexpress the relaxed  SRM problem \eqref{SRMRe3} as
\begin{align}\label{AppenB1}
\tilde{R}_S^{\star}=\underset{0\le \tau \le 1,\eta \ge 1}{\text{max}}\left\{\begin{array}{l}
   \underset{\substack{\widetilde{ \bf W}_P\succeq{ \bf   0},
  \widetilde{ \bf P}_{F}\succeq{ \bf   0}}}
{\text{max}}~ { C_1{ \bf   h}_{R}^H
   \widetilde{ \bf   P}_{F}
   { \bf   h}_{R}}\\
     ~~~~~{\rm{s.t.}}~~~~
 ~~\widetilde{\text{CR1}}\sim \widetilde{\text{CR4}}  \end{array}\right..
\end{align}
Then the  proof  for Theorem~\ref{thre3} consists of two steps:  In the first step, we prove that  for any given  $\tau$ and $\eta$, the optimal solution $\{ {\widetilde{\bf  W}}_{P, 1}, \widetilde{ \bf P}_{F,1}\}$ to the inner maximization  problem  in \eqref{AppenB1} is of  rank-1. In the second step, we show that   $\{ {\widetilde{\bf  W}}_{P,1}, \widetilde{ \bf P}_{F,1}\}$ is also  optimal to   problem \eqref{SRM2} ( equivalent to the perfect SRM problem  \eqref{SRM}),  when the same  $\tau$ and $\eta$ are adopted. As a result, the equivalence between   problem \eqref{SRMRe3}   and the   perfect SRM problem  \eqref{SRM} is established.

\textbf{Step 1:}
To prove that the optimal solution to  the inner maximization problem in \eqref{AppenB1} is of  rank-1, we firstly consider its equivalent counterpart, i.e. the following power minimization problem 
\begin{align}\label{SP0}
 &  \underset{\substack{\widetilde{ \bf W}_P\succeq{ \bf   0},
  \widetilde{ \bf P}_{F}\succeq{ \bf   0}}}
{\text{min}}~\text{tr}(\widetilde{ \bf   P}_{F}),~~~{\rm{s.t.}}~~~\widetilde{\text{CR1}}\sim \widetilde{\text{CR4}}; ~~\widetilde{
      \text{PR1}}:  C_1{ \bf   h}_{R}^H
   \widetilde{ \bf   P}_{F}
   { \bf   h}_{R}\ge f_{\eta,\tau}
\end{align}
where $ f_{\eta,\tau}$ denotes the maximum  objective  value of the inner maximization problem in \eqref{AppenB1}. Clearly,  problem \eqref{SP0} is still convex with the linear objective function and  convex constraints. In fact, the optimal solution of the   inner maximization problem in \eqref{AppenB1}  is the same as that of  problem \eqref{SP0},  which can be proved by  contradiction.
Firstly,
 by assuming  that the optimal solution  to the inner maximization problem in \eqref{AppenB1} and  problem \eqref{SP0} are $\{\widetilde{\bf  W}_{P,1}, \widetilde{ \bf P}_{F,1}\}$ and $\{ \widetilde{\bf  W}_{P,2}, \widetilde{ \bf P}_{F,2}\}$, respectively,  we  readily    find that
$( \widetilde{ \bf W}_{P,2},\widetilde{ \bf P}_{F,2} )$   is  feasible  for the inner maximization problem in \eqref{AppenB1} and thus we have
$
{ C_1{ \bf   h}_{R}^H
  \widetilde{ \bf   P}_{F,2}
  { \bf   h}_{R}}\le{ C_1{ \bf   h}_{R}^H
  \widetilde{ \bf   P}_{F,1}
  { \bf   h}_{R}}=f_{\eta,\tau}$. Furthermore, since $( \widetilde{ \bf W}_{P,2},\widetilde{ \bf P}_{F,2} )$  also satisfies $\widetilde{
  \text{PR1}}$, it is inferred that
   $ C_1{ \bf   h}_{R}^H
    \widetilde{ \bf   P}_{F,2}
    { \bf   h}_{R}=f_{\eta,\tau}$.
  Clearly, the optimal solution $( \widetilde{ \bf W}_{P,2},\widetilde{ \bf P}_{F,2} )$  to  problem \eqref{SP0} is also optimal to the inner maximization problem in \eqref{AppenB1}.  In the sequel, we instead  investigate  the  rank-1 property of  the optimal solution  to    problem  \eqref{SP0}  through Karush-Kuhn-Tucker
(KKT) conditions, which are shown  in \eqref{AppenB1} at the top of this page.
\begin{figure*}
  \begin{subequations}\label{AppenB1}
  \vspace{-10mm}
    \begin{align}
  &(1+\beta^{\star}){\bf I}_{N_F}+\gamma^{\star}{\bf g}_P{\bf g}_P^H+
\rho^{\star}{\bf H}_E^H{\bf R}_E{\bf H}_E-\psi^{\star}C_1{\bf h}_R{\bf  h}_R^H-{\bf Z}_F^{\star}={\bf 0}\label{A1}\\
  &\lambda^{\star} {\bf I}_{N_P}+\gamma^{\star}{\bf h}_{M}{\bf h}_{M}^H-\beta^{\star}\xi
  {\bf H}_F^H{\bf H}_F-{\bf Z}_P^{\star}={\bf 0}\label{A2}\\
  &{\bf Z}_F^{\star}\widetilde{ \bf P}_{F}^{\star}={\bf 0},~{\bf Z}_P^{\star}
  \widetilde{ \bf W}_{P}^{\star}={\bf 0}\label{A3}\\
  &\lambda^{\star}(\text{tr}(\widetilde{ \bf   W}_P^{\star}) -\tau P_P)=0\label{A4}\\
  &\beta^{\star}(\text{tr}(\widetilde{ \bf   P}_{F}^{\star})\!
-\! \xi \text{tr}({ \bf    H}_F\widetilde{ \bf    W}_P^{\star}{ \bf    H}_F^H)-\xi \tau\frac{P_M}{N_M}\text{tr}(
  { \bf    G}_F{ \bf    G}_F^H))=0\label{A5}\\
  &\gamma^{\star}({ \bf   h}_{M}^H \widetilde{ \bf   W}_P^{\star}{ \bf   h}_{M} +{ \bf g}_{p}^H \widetilde{ \bf P}_{F}^{\star}{ \bf g}_{p}- I_{th})=0\label{A6}\\
  &\rho^{\star}((1-\eta)(1-\tau)+\text{tr}\big( {\bf R}_E{ \bf   H}_{E}
  \widetilde { \bf   P}_{F}^{\star}{ \bf   H}_{E}^H\big))=0\label{A7}\\
  &\psi^{\star}(   C_1{ \bf   h}_{R}^H
  \widetilde{ \bf   P}_{F}^{\star}
  { \bf   h}_{R}-f_{\eta,\tau})=0\label{A8}
\end{align}  
  \end{subequations}
  \vspace{-4mm}
\hrulefill
\vspace{-6mm}
\end{figure*}
 Here,
  $ \{\lambda, \beta, \gamma, \rho,\psi  \}$   denote
 non-negative lagrangian multipliers  for
 constraints $\{\widetilde{\text{CR}}1$,$\widetilde{\text{CR}}2$,
$\widetilde{\text{CR}}3$, $\widetilde{\text{CR}}4,\widetilde{\text{PR}}1\}$ in  problem \eqref{SP0}, respectively. While $\bm{Z}_P\succeq\bm{0}$ and  $\bm{Z}_F\succeq\bm{0}$ are the lagrangian multipliers corresponding to
$\widetilde{\bm{W}}_P\succeq\bm{0}$ and $\widetilde{\bm{P}}_{F}\succeq\bm{0}$, respectively.  Based on \eqref{A1} and \eqref{A3},
we have%
  $\left((1 \!+ \!\beta^{\star}){\bf I}_{N_F}\!+\!\gamma^{\star}{\bf g}_P{\bf g}_P^H \!+\! 
\rho^{\star}{\bf H}_E^H{\bf R}_E{\bf H}_E\right)\widetilde{ \bf P}_{F}^{\star}\!=\!\psi^{\star}C_1{\bf h}_R{\bf  h}_R^H\widetilde{ \bf P}_{F}^{\star},$  implying  that
\begin{align}\label{A9}
 &\text{rank}\big(\big((1+\beta^{\star}){\bf I}_{N_F}+\gamma^{\star}{\bf g}_P{\bf g}_P^H+
\rho^{\star}{\bf H}_E^H{\bf R}_E{\bf H}_E\big)\widetilde{ \bf P}_{F}^{\star}\big)= \text{rank}(\widetilde{ \bf P}_{F}^{\star}) \nonumber\\
&=\text{rank}(\psi^{\star}C_1{\bf h}_R{\bf  h}_R^H\widetilde{ \bf P}_{F}^{\star})\le 1
\end{align} where the first  equality in \eqref{A9} is due to
$(1+\beta^{\star}){\bf I}_{N_F}+\gamma^{\star}{\bf g}_P{\bf g}_P^H+
\rho^{\star}{\bf H}_E^H{\bf R}_E{\bf H}_E\succ {\bf 0}$. According to \eqref{A9},  the  rank-1 optimal $\widetilde{ \bf   P}_{F}^{\star}$ to  problem \eqref{SP0} is proved.

As for the optimal $\widetilde{ \bf   W}_P^{\star}$,  we refer to \eqref{A2} and \eqref{A3} to  obtain
\begin{align}\label{A10}
  \lambda^{\star} {\bf I}_{N_P}-\widetilde{\bf H}^{\star} ={\bf Z}_P^{\star}\succeq {\bf 0}, ~ (\lambda^{\star} {\bf I}_{N_P}-\widetilde{\bf H}^{\star} )\widetilde{ \bf   W}_P^{\star}={\bf 0},
\end{align} where $ \widetilde{\bf H}^{\star}=\beta^{\star}\xi
  {\bf H}_F^H{\bf H}_F-\gamma^{\star}{\bf h}_{M}{\bf h}_{M}^H$.  Observing from  \eqref{A10}, if $\lambda^{\star}=0$, then we have $\beta^{\star}\xi
  {\bf H}_F^H{\bf H}_F\preceq\gamma^{\star}{\bf h}_{M}{\bf h}_{M}^H$, for which  $\beta^{\star}=0$ is derived  since  $\text{rank}(\xi
  {\bf H}_F^H{\bf H}_F)>\text{rank}({\bf h}_{M}{\bf h}_{M}^H) =1$ is implied by   our system setting. The optimal $\widetilde{ \bf   W}_P^{\star}$ is then derived as ${\bf h}_{M}{\bf h}_{M}^H\widetilde{ \bf   W}_P^{\star}=0$.
  Clearly, there must exist a rank-1 optimal $\widetilde{ \bf   W}_P^{\star}$ in the null space of ${\bf h}_{M}$. However, if  $\lambda^{\star}>0$,  then we   have
    the eigenvalue decomposition (EVD) of $\widetilde{\bf H}^{\star}$ as   $ \widetilde{\bf H}^{\star}\!\!=\!\!
  \widetilde{\bf U}_{H}\widetilde{\bf \Lambda
  }_{H}\widetilde{\bf U}_{H}^H$,  where
the maximum eigenvalue  $\tilde{\lambda}_{{{H}},\max}$ in the   diagonal matrix $\widetilde{\bf \Lambda
  }_{{{H}}}$  must be   positive.    Otherwise, $\beta^{\star}\xi
  {\bf H}_F^H{\bf H}_F \preceq \gamma^{\star}{\bf h}_{M}{\bf h}_{M}^H$  is implied and  only obtained at $\beta^{\star}=0$ as mentioned above.
Based on \eqref{A10}, we  then have  $\widetilde{ \bf   W}_P^{\star}={\bf 0}$  since $\lambda^{\star}>0$, which contradicts with \eqref{A4}.  So  $\lambda^{\star}>0$ yields $\tilde{\lambda}_{{{H}},\max}>0$ and thus we have  $ \lambda^{\star} {\bf I}_{N_P}-\widetilde{\bf H}^{\star}=  \widetilde{\bf U}_{H}(\lambda^{\star} {\bf I}_{N_P}-\widetilde{\bf \Lambda
  }_{H})\widetilde{\bf U}_{H}^H\succeq{\bf 0}$.   Further, to ensure
$\widetilde{ \bf   W}_P^{\star}\neq {\bf 0}$ in \eqref{A10},  it is easily inferred that  $\lambda^{\star} =\tilde{\lambda}_{{{H}},\max}$ and thus $\widetilde{ \bf   W}_P^{\star}=c{\bf u}_P{\bf u}_P^H$ is obtained,  where   ${\bf u}_{P}$ is the unit-norm  eigenvector of $\widetilde{\bf H}^{\star} $  corresponding to $\lambda_{{{H}},\max}$. Moreover, since $\lambda^{\star}=\tilde{\lambda}_{{{H}},\max}>0$, it yields  that $\text{tr}(\widetilde{ \bf   W}_P^{\star}) =\tau P_P$  and  thus $\widetilde{ \bf   W}_P^{\star}=\tau P_P{\bf u}_P{\bf u}_P^H$ can be obtained from   \eqref{A4}.

Both  cases of $\lambda^{\star}$ demonstrate that the  optimal $\widetilde{ \bf   W}_P^{\star}$ to  problem   \eqref{SP0} is  of rank-1.  Overall,  the rank-1 optimal solution $\{\widetilde{ \bf W}_{P},
 \widetilde{ \bf P}_{F}\}$  of 
the inner maximization problem in \eqref{AppenB1} is proved.

\textbf{Step 2:}
 Firstly, we assume  that the optimal solution of  the inner maximization problem in \eqref{AppenB1} (equivalent to problem ) and  problem \eqref{SRM2} with the same  $\{\eta,\tau\}$  are   $\{\widetilde{ \bf W}_{P,1},
\widetilde{ \bf P}_{F,1}\}$ and $\{\widetilde{ \bf W}_{P,0},
 \widetilde{ \bf P}_{F,0}\}$, respectively. The corresponding   objective function  is defined as  $  f(\widetilde{ \bf W}_{P},
 \widetilde{ \bf P}_{F}) $.  Based on Lemma~\ref{lemm0}, for any given $\eta$ and $\tau$,    the    relaxed SRM problem  \eqref{SRMRe3}  actually  has a larger feasible solution region  than  problem \eqref{SRM2}, hence, 
 we    have
  $f(\widetilde{ \bf W}_{P,1},
 \widetilde{ \bf P}_{F,1}) \ge f(\widetilde{ \bf W}_{P,0},
 \widetilde{ \bf P}_{F,0})$.  Furthermore,
 since  both
$\widetilde{ \bf W}_{P,1},
 $ and $\widetilde{ \bf P}_{F,1}$  are of rank-1 as proved in Step 1, we also find that   $\{ \widetilde{ \bf W}_{P,1},
 \widetilde{ \bf P}_{F,1}\}$  is  feasible to  problem  \eqref{SRM2}, which implies   that
  $f(\widetilde{ \bf W}_{P,1},
 \widetilde{ \bf P}_{F,1}) \le f(\widetilde{ \bf W}_{P,0},
 \widetilde{ \bf P}_{F,0})$.
Combining the above two inequalities, we  finally have
  $f(\widetilde{ \bf W}_{P,1},
 \widetilde{ \bf P}_{F,1}) = f(\widetilde{ \bf W}_{P,0},
 \widetilde{ \bf P}_{F,0})$.

Given any  $\tau$ and $\eta$, it is concluded that   such a two-tupple  $\{ {\widetilde{\bf  W}}_{P,1}, \widetilde{ \bf P}_{F,1}\}$  from the relaxed SRM  problem \eqref{SRMRe3} (inner maximization   problem in  \eqref{AppenB1}) is   optimal to   problem \eqref{SRM2}. Notice that  problem \eqref{SRM2} with  the optimal  $\tau^{\star}$ and $\eta^{\star}$   is equivalent to the original SRM problem \eqref{SRM},    Therefore,   the  equivalence between  problem \eqref{SRMRe3} and the PSRM problem \eqref{SRM}  is  established as that in Theorem~\ref{thre3}.

\vspace{-3mm}
 \section{}\label{AppenB}
Notice that  this proof follows the same logic as that  provided for Theorem~\ref{thre3}.  Firstly, assuming the fixed  $\tau$ and  $\eta$,  we can  reexpress the relaxed robust SRM problem \eqref{WSRM5} as
 \begin{align}\label{AWSRM5}
  &\tilde{R}_{S,Ro}^{\star}=
 \underset{0\le \tau \le 1,\eta \ge 1}{\text{max}}\left\{\begin{array}{l}
  \underset{\widetilde{ \bf W}_P\succeq{ \bf   0},
  \widetilde{ \bf P}_{F}\succeq{ \bf   0} }
{\text{max}}~~{C_1 { \bf   h}_{R}^H
  \widetilde { \bf   P}_{F}
   { \bf   h}_{R}} \\
  {\rm{s.t.}}~\widetilde{\text{CR1}}\sim \widetilde{\text{CR3}}\\~~~~~
  \widetilde{\text{CR5}}:~\widetilde{\bf I}_E-\widetilde{\bf R}_E({\bf I}_{N}\otimes \widetilde{\bf P}_F)\widetilde{\bf R}_E^H\succeq{ \bf   0}
\end{array}\right.
\end{align} where
$\widetilde{\bf I}_E= \left[\begin{array}{cc}
  u{ \bf   I}_{N} & {\bf 0}\\
{\bf 0}& \eta-1-\mu\xi_f^2\end{array}\right]$ and $\widetilde{\bf R}_E= \left[\begin{array}{c}
  {\bf R}_E^{\frac{1}{2} T}\otimes { \bf   I}_{N} \\
  \hat{ \bf   h}_E^H({\bf R}_E^{\frac{1}{2} T}\otimes { \bf   I}_{N})
\end{array}\right].$  The  proof of Theorem~\ref{thre4} is similar to that of  Theorem~\ref{thre3} and  also consists of the following two steps.

In the first step, we prove that  for any given  $\tau$ and $\eta$, the rank-1 optimal  solution $\{ {\widetilde{\bf  W}}_{P}, \widetilde{ \bf P}_{F}\}$ of the inner maximization  problem  in \eqref{AWSRM5} is obtained. In the second step, we show that  such an    $\{ {\widetilde{\bf  W}}_{P}, \widetilde{ \bf P}_{F}\}$ is also  optimal to the  WSRM problem  \eqref{WSRM}. As a result,  the equivalence between  the relaxed WSRM problem \eqref{WSRM5} and  the original WSRM problem \eqref{WSRM} are  established.

Similarly, for   proving   the rank-1 nature of the  optimal solution to  the inner maximization problem in \eqref{AWSRM5}, we  consider the corresponding      power minimization problem  given $\tau$ and $\eta$  as
\begin{align}\label{SPR1}
 &  \underset{\substack{\widetilde{ \bf W}_P\succeq{ \bf   0},
  \widetilde{ \bf P}_{F}\succeq{ \bf   0}}}
{\text{min}}~\text{tr}(\widetilde{ \bf   P}_{F}),~~~{\rm{s.t.}}~~~\widetilde{\text{CR1}}\sim \widetilde{\text{CR3}};~~\widetilde{\text{CR5}};~~\widetilde{
      \text{PR}}1:  C_1{ \bf   h}_{R}^H
   \widetilde{ \bf   P}_{F}
   { \bf   h}_{R}\ge \hat{f}_{\eta,\tau}., 
\end{align}
where $ \hat{f}_{\eta,\tau}$ denotes the  optimal  objective value of the inner maximization problem in \eqref{AWSRM5}  given  any  $\tau$ and $\eta$.  It is readily observed that   problem \eqref{SPR1} is also convex. Following the same philosophy  in  Step 1 of  Appendix~\ref{AppenA}, we  readily  verify that  the optimal solution of  problem \eqref{SPR1} is also optimal to the  inner maximization problem in \eqref{AWSRM5}. In the sequel,  we  aim to show that the optimal
solution $\{\widetilde{ \bf W}_P,
\widetilde{ \bf P}_{F}\}$ of  problem \eqref{SPR1} for any fixed  $\tau$ and $\eta$  are of rank-1 through KKT conditions, which are
   \begin{subequations}\label{AppenB2}
     \begin{align}
   &(1+\beta^{\star}){\bf I}_{N_F}+\gamma^{\star}{\bf g}_P{\bf g}_P^H+
 \rho^{\star}{\bf H}_E^H{\bf R}_E{\bf H}_E+\sum\limits_{i=1}^{N_E}\widetilde{\bf R}_{E,i}^H{\bf Z}_R \widetilde{\bf R}_{E,i}-\psi^{\star}C_1{\bf h}_R{\bf  h}_R^H-{\bf Z}_F^{\star}={\bf 0}\label{AA0}\\
   &{\bf Z}_R^{\star}(\widetilde{\bf I}_E-\widetilde{\bf R}_E({\bf I}_{N}\otimes \widetilde{\bf P}_F)\widetilde{\bf R}_E^H)={\bf 0}\label{AA2}\\
   &\eqref{A2}\sim \eqref{A6},~\eqref{A8}
 \end{align}
   \end{subequations}
  where  $\bm{Z}_R\succeq\bm{0}$  is   lagrangian multiplier   for  $\widetilde{\text{CR5}}$. Based on
   \eqref{AA0} and  \eqref{A3}, we have
 \begin{align}\label{AA9}
  &\text{rank}\big(\big((1+\beta^{\star}){\bf I}_{N_F}+\gamma^{\star}{\bf g}_P{\bf g}_P^H+
 \rho^{\star}{\bf H}_E^H{\bf R}_E{\bf H}_E+\sum\limits_{i=1}^{N_E}\widetilde{\bf R}_{E,i}^H{\bf Z}_R \widetilde{\bf R}_{E,i}\big)\widetilde{ \bf P}_{F}^{\star}\big)= \text{rank}(\widetilde{ \bf P}_{F}^{\star}) \nonumber\\
 &=\text{rank}(\psi^{\star}C_1{\bf h}_R{\bf  h}_R^H\widetilde{ \bf P}_{F}^{\star})\le 1.
 \end{align} 
 
 It is clear from \eqref{AA9} that the optimal $\widetilde{ \bf P}_{F}^{\star}$ to  problem \eqref{SPR1} is of rank-1. In addition,  since the $\widetilde{ \bf W}_P$ related KKT conditions of  problem \eqref{SPR1} are the same as that of  problem \eqref{SP0}. 
 Let's refer to  Step 1 of Appendix C  to prove  the rank-1  optimal $\widetilde{ \bf W}_P^{\star}$ to  problem \eqref{SPR1}. As such, the  rank-1 optimal solution $\{\widetilde{ \bf   P}_F^{\star},\widetilde{ \bf   W}_P^{\star}\}$ to  problem   \eqref{SPR1} are finally verified.
\vspace{-2mm}
\section{}\label{AppenC}
According to \cite{Ma:2014kzbacada} and following the same argument  as  Step 1 of Appendix~\ref{AppenA},   the outage-constrained SRM
 problem \eqref{OUTSRM2} given any $\{\tau,{\bf S}\}$  is   equivalent to the following power minimization problem
  \begin{align}\label{OUTSRM2A}
    &\underset{\substack{\widetilde{ \bf   W}_P\succeq{ \bf   0}, \widetilde{ \bf   P}_{F}
      \succeq{ \bf   0}}} {\text{min}}~~~~~\text{tr}(\widetilde{ \bf   P}_{F}),~~{\rm{s.t.}}~~\widetilde{\text{CR1}}\sim
    \widetilde{\text{CR3}},~\widetilde{\text{CR6}}|_{ R_S=R_S^{1}},
    \end{align}  where $R_S^{1}$ denotes  the  optimal objective value of the problem \eqref{OUTSRM2} with fixed  $\tau$ and ${\bf S} $.
In other words, the  optimal solution  $\{\widetilde{ \bf   P}_F^{\star},\widetilde{ \bf   W}_P^{\star}\}$ to   problem \eqref{OUTSRM2A} is also optimal to  problem \eqref{OUTSRM2} given the same  $\tau$ and ${\bf S}$.
In the sequel,  we firstly prove  the  rank-1 optimal $\widetilde{ \bf   P}_F^{\star}$  to   problem \eqref{OUTSRM2A}. First of all, we define the  projection matrix ${\bf T}$ of the  vector $\widetilde{ \bf   P}_F^{\frac{1}{2} \star}{\bf h}_R$ as
  ${\bf T} = \frac{\widetilde{ \bf   P}_F^{\frac{1}{2} \star}{\bf h}_R{\bf h}_R^H\widetilde{ \bf   P}_F^{\frac{1}{2} \star}}{\Vert{\bf h}_R^H\widetilde{ \bf   P}_F^{ \frac{1}{2} \star}\Vert^2}$, then
  a novel rank-1  solution
$\widehat{ \bf   P}_F^{\star} = \widetilde{ \bf   P}_F^{\frac{1}{2} \star}{\bf T} \widetilde{ \bf   P}_F^{\frac{1}{2} \star}$   is  available and  satisfies
\begin{align}\label{ae2}
 \text{tr}(\widehat{ \bf   P}_F^{\star})-
 \text{tr}(\widetilde{ \bf   P}_F^{\star})=
 \text{tr}(\widetilde{ \bf   P}_F^{\frac{1}{2} \star}({\bf T}-{\bf I}) \widetilde{ \bf   P}_F^{\frac{1}{2} \star})\le 0.
\end{align} 
The formulation \eqref{ae2} yields    $\text{tr}(\widehat{ \bf   P}_F^{\star})\le
\text{tr}(\widetilde{ \bf   P}_F^{\star})$, which hints that  the objective value of problem \eqref{OUTSRM2}  is non-increasing   while still satisfying  constraints $\widetilde{\text{CR1}}\sim
\widetilde{\text{CR2}}$. Moreover, observing from $\widetilde{\text{CR6}}$, we
 have
\begin{align}\label{ae3}
 \log_2\big(1+ \frac{C_1 {\bf   h}_{R}^H
 \widehat{ \bf   P}_{F}^{\star}{ \bf   h}_{R} }{1-\tau}
\big)
=\log_2\bigg(1+\frac{C_1\vert{ \bf   h}_{R}^H
\widetilde{ \bf   P}_F^{\star}{\bf h}_R\vert^2}{(1-\tau)\Vert\widetilde{ \bf   P}_F^{ \frac{1}{2} \star}{\bf h}_R \Vert^2}\bigg)=\log_2\big(1+ \frac{C_1 { \bf   h}_{R}^H
\widetilde{ \bf   P}_{F}^{\star}{ \bf   h}_{R}}{1-\tau}
\bigg)
\end{align} as well as
{
    \begin{align}\label{ae4}
     & \log_2\det\bigg({{ \bf   I}}_{N_E}\!+\!\frac{ {\bf R}_E{ \bf   H}_{E}
      \widehat{ \bf   P}_{F}^{\star}{ \bf   H}_{E}^H}{1-\tau}\bigg)
     =\log_2\bigg({{ \bf   I}}_{N_E}\!+\!\frac{ {\bf R}_E{ \bf   H}_{E}\widetilde{ \bf   P}_F^{\frac{1}{2} \star}{\bf T}\widetilde{ \bf   P}_F^{\frac{1}{2} \star}{ \bf   H}_{E}^H}{(1-\tau)}\bigg)\nonumber\\
     &=\log_2\bigg(1\!+\!\frac{ {\bf h}_R^H\widetilde{ \bf   P}_F^{\frac{1}{2} \star}(\widetilde{ \bf   P}_F^{\frac{1}{2} \star}{ \bf   H}_{E}^H{\bf R}_E{ \bf   H}_{E}\widetilde{ \bf   P}_F^{\frac{1}{2} \star})\widetilde{ \bf   P}_F^{\frac{1}{2} \star}{\bf h}_R}{(1-\tau)\Vert{\bf h}_R^H\widetilde{ \bf   P}_F^{ \frac{1}{2} \star}\Vert^2}\bigg)\le \log_2\big(1\!+\frac{\lambda_{\max}(\widetilde{ \bf   P}_F^{\frac{1}{2} \star}{ \bf   H}_{E}^H{\bf R}_E{ \bf   H}_{E}\widetilde{ \bf   P}_F^{\frac{1}{2} \star})}{1-\tau}\big)\nonumber\\
     &\le \log_2\det({\bf I }_{N_F}+\frac{\widetilde{ \bf   P}_F^{\frac{1}{2} \star}{ \bf   H}_{E}^H{\bf R}_E{ \bf   H}_{E}\widetilde{ \bf   P}_F^{\frac{1}{2} \star}}{1-\tau})=\log_2\det({\bf I }_{N_E}+\frac{{\bf R}_E{ \bf   H}_{E}\widetilde{ \bf   P}_F^{\star}{ \bf   H}_{E}^H}{1-\tau}).
    \end{align}
}
Based on \eqref{ae3} and \eqref{ae4},  it yields
 \begin{align}\label{ae40}
  & \log_2\big(1+ \frac{C_1 {\bf   h}_{R}^H
  \widehat{ \bf   P}_{F}^{\star}{ \bf   h}_{R} }{1-\tau}
 \big)-\log_2\det\bigg({{ \bf   I}}_{N_E}\!+\!\frac{ {\bf R}_E{ \bf   H}_{E}
   \widehat{ \bf   P}_{F}^{\star}{ \bf   H}_{E}^H}{1-\tau}\bigg)
  \nonumber\\
  &\ge \log_2\big(1+ \frac{C_1 {\bf   h}_{R}^H
  \widetilde{ \bf   P}_{F}^{\star}{ \bf   h}_{R} }{1-\tau}
 \big)-\log_2\det\bigg({{ \bf   I}}_{N_E}\!+\!\frac{ {\bf R}_E{ \bf   H}_{E}
  \widetilde{ \bf   P}_{F}^{\star}{ \bf   H}_{E}^H}{1-\tau}\bigg),
 \end{align}  which means that  the constraint   $\widetilde{\text{CR6}}$ still holds by using the novel  solution $\widehat{ \bf   P}_F^{ \star}$. Similarly, by referring to \eqref{ae2} and  the following inequality
\begin{align}\label{ae5}
 {\bf g}_p^H \widehat{ \bf   P}_F^{\star} {\bf g}_p=
 {\bf g}_p^H \widetilde{ \bf   P}_F^{\frac{1}{2} \star}{\bf T} \widetilde{ \bf   P}_F^{\frac{1}{2} \star} {\bf g}_p=\frac{\vert {\bf g}_p^H \widetilde{ \bf   P}_F^{ \frac{1}{2} \star}\widetilde{ \bf   P}_F^{ \frac{1}{2} \star}{\bf h}_R\vert^2  } { \Vert\widetilde{ \bf   P}_F^{\frac{1}{2} \star}{\bf h}_R\Vert^2}\le {\bf g}_p^H\widetilde{ \bf   P}_F^{ \star}{\bf g}_p,
\end{align}
the constraint  $\widetilde{\text{CR3}}$ can also be satisfied with  $\{\widehat{ \bf   P}_F^{ \star}, \widetilde{ \bf   W}_P^{\star}\}$.  This phenomenon  indicates that the
novel  solution $\{\widehat{ \bf   P}_F^{ \star}, \widetilde{ \bf   W}_P^{\star}\}$ is also  feasible to  problem \eqref{OUTSRM2A} and  may even realize a  lower objective value than $\{\widetilde{ \bf   P}_F^{ \star}, \widetilde{ \bf   W}_P^{\star}\}$. As a result,  we  conclude that the  optimal ${ \bf   P}_F^{ \star}$ to   problem \eqref{OUTSRM2A}  (equivalent to the outage-constrained SRM problem \eqref{OUTSRM2}) must be of rank-1. Since  the $\widetilde{ \bf W}_P$ related KKT conditions of problem \eqref{OUTSRM2} are also  the same as that of  problem \eqref{SP0},  the proof for    the rank-1  optimal $\widetilde{ \bf W}_P$ to  problem \eqref{SP0} in  Appendix~\ref{AppenA}
can still  be applied to    problem \eqref{OUTSRM2}. Due to   space limitation, the detailed proof is omitted here. Overall, the  rank-1 optimal solution $\widetilde{ \bf   P}_F^{\star}, \widetilde{ \bf   W}_P^{\star}$ to the outage-constrained SRM problem \eqref{OUTSRM2} is   proved.
\vspace{-3mm}
\section{}\label{AppenD}
Recalling the proof in Step 1 of Appendix A,   for any fixed $\tau$,  we readily infer that  the  AN aided problem \eqref{ARSRM3}  is  equivalent to the following  power minimization problem
\begin{align}\label{ARSRM4}
  &\underset{\substack{{\tau},\widetilde{ \bf   W}_P\succeq{ \bf   0},
  \widetilde{ \bf   P}_{F}\succeq{ \bf   0},\widetilde{ \bf   \Sigma}_U\succeq{ \bf   0} }}
{\text{min}}~~\text{tr}( \widetilde{{ \bf   P}}_{F})\nonumber\\
&~~~~~~~~~{\rm{s.t.}}~
\widetilde{\text{AR1}}:\text{tr}( \widetilde{ \bf   \Sigma}_{U}) \ge f_{A},~\widetilde{\text{AR2}}: \text{tr}(C_1{ \bf   h}_{R}^H
  \widetilde { \bf   P}_{F}
   { \bf   h}_{R}) \ge C_2,~\widetilde{\text{CR1}},~\widetilde{\text{AR3}},~\widetilde{\text{CR3}}, 
\end{align} where  $f_A$ denotes  the optimal objective value of  problem \eqref{ARSRM3} and $C_2=(1-\tau)(2^{\frac{R_{th}}{(1-\tau)}}-1)$.
To demonstrate  the  rank-1 optimal solution $\{\widetilde{ \bf   W}_P,
\widetilde{ \bf   P}_{F},\widetilde{ \bf   \Sigma}_U \}$ to  problem \eqref{ARSRM4}, we formulate the corresponding
KKT conditions as
\begin{subequations}\label{AppenB3}
  \begin{align}
    &(1+\beta^{\star}){\bf I}_{N_F}+\gamma^{\star}{\bf g}_p{\bf g}_p^H-\psi^{\star}C_1{\bf h}_R{\bf h}_R^H-{\bf Z}_F^{\star}={\bf 0},\label{eqq1}\\
    &(\beta^{\star}-\rho^{\star}){\bf I}_{N_F}+\gamma^{\star} {\bf g}_p{\bf g}_p^H-{\bf Z}_{U}^{\star}={\bf 0},\label{eqq2}\\
&{\bf Z}_U^{\star}\widetilde{ \bf   \Sigma}_U^{\star} ={\bf 0},~\eqref{A2}\sim \eqref{A6},~\eqref{A8}|_{ f_{\eta,\tau}=C_2},\label{eqq3}
\end{align}
\end{subequations}
where $\{\rho^{\star}, \psi^{\star},
\beta^{\star}\}$, $\lambda^{\star}$ and $\gamma^{\star}$ are the optimal lagrangian multipliers
associated with  constraints $\widetilde{\text{AR1}}\sim
 \widetilde{\text{AR3}}$,   $\widetilde{\text{CR1}}$ and $\widetilde{\text{CR3}}$, respectively, while  ${\bf Z}_U^{\star} \succeq {\bf 0}$ is the optimal lagrangian multiplier for  $\widetilde{\bf \Sigma}_U \succeq {\bf 0}$.
It is readily observed from \eqref{eqq1}  and  \eqref{A3}  that
 \begin{align}\label{eqq4}
  &\text{rank}\big(\big[(1+\beta^{\star}){\bf I}_{N_F}+\gamma^{\star}{\bf g}_P{\bf g}_P^H\big]\widetilde{ \bf P}_{F}^{\star}\big)= \text{rank}(\widetilde{ \bf P}_{F}^{\star})=\text{rank}(\psi^{\star}C_1{\bf h}_R{\bf  h}_R^H\widetilde{ \bf P}_{F}^{\star})\le 1,
\end{align}
 which implies that the optimal $\widetilde{ \bf P}_{F}^{\star}$ to  problem \eqref{ARSRM4} is of rank-1.
Additionally, based on  \eqref{eqq2} and \eqref{eqq3},  we also find that $\beta^{\star}-\rho^{\star}\ge 0$   for guaranteeing ${\bf Z}_U^{\star} \succeq{\bf 0}$. To be specific, when $\beta^{\star}-\rho^{\star}= 0$, we have  $\gamma^{\star} {\bf g}_p{\bf g}_p^H \widetilde{\bf \Sigma}_U^{\star}={\bf 0}$. Clearly,  there  must exist a  rank-1 optimal
$\widetilde{\bf \Sigma}_U^{\star}$ within the null space of ${\bf g}_p^H $. While for  $\beta^{\star}-\rho^{\star}>0$,   ${\bf Z}_U^{\star} \succ {\bf 0}$ and   $\widetilde{\bf \Sigma}_U^{\star}={\bf 0}$ are obtained  according to \eqref{eqq3}. Based on above discussion, the rank-1 optimal  $\widetilde{\bf \Sigma}_U^{\star}$ to  problem \eqref{ARSRM4}  can be proved. Without loss of generality, the $\widetilde{ \bf   W}_P$ related constraints of  problem \eqref{ARSRM4} are also identical to that of  problem  \eqref{SP0}. Therefore, the rank-1 nature of the optimal $\widetilde{ \bf   W}_P^{\star}$ to  problem  \eqref{ARSRM4} can  be  verified. Considering the equivalence between  problems \eqref{ARSRM4}  and  \eqref{ARSRM3}, we  finally conclude  that the optimal solution $\{\widetilde{ \bf   W}_P,
\widetilde{ \bf   P}_{F},\widetilde{ \bf   \Sigma}_U \}$ of  problem \eqref{ARSRM3} is also of rank-1.

\section{}\label{AppenE}
Since it has been proved that the globally optimal or high-quality suboptimal  FBS information  covariance matrices  ${ \bf P}_F$   for all studied  problems, namely \eqref{SRM}, \eqref{WSRM}, \eqref{OUTRM} and \eqref{ARSRM2}, are all  of rank-1, accordingly, we   can define the optimal  ${\bf P}_F^{\star}= \lambda_P{\bf p}_F{\bf p}_F^H$  with $\Vert {\bf p}_F\Vert_F=1$  and  $\lambda_P>0$   for   all  these   problems.
Furthermore,   by substituting   ${\bf P}_F^{\star}= \lambda_P{ \bf p}_F{\bf p}_F^H$ into the secrecy rate expression    in \eqref{RS1}, it yields
\begin{align}\label{WPno0}
   R_S &=(1-\tau)\log_2\big(1+ C_1 \lambda_P{ { \bf   h}_{R}^H
 \bm{p}_F\bm{p}_F^H
   { \bf   h}_{R}}\big)- (1-\tau)\log_2\det\big({{ \bf   I}}_{N_E}\!+\!  {\bf R}_E{ \bf   H}_{E}
   { \bf   P}_{F}{ \bf   H}_{E}^H\big)\nonumber\\
   &=(1-\tau)\log_2 \bigg(1+\frac{ C_1 \vert { \bf   h}_{R}^H
 \bm{p}_F\vert^2- \Vert   {\bf R}_E^{\frac{1}{2}} { \bf   H}_{E}
   { \bf  p}_{F}\Vert_F^2 }{   \frac{1}{ \lambda_P}+    \Vert   {\bf R}_E^{\frac{1}{2}} { \bf   H}_{E}
   { \bf  p}_{F}\Vert_F^2  }\bigg).
\end{align} 

Notice that  our work aims to   optimize  $\bm{P}_F$   for achieving the  maximum  non-negative secrecy rate $   R_S$,   so we must have $ C_1 \vert { \bf   h}_{R}^H
 \bm{p}_F \vert^2 \ge \Vert   {\bf R}_E^{\frac{1}{2}} { \bf   H}_{E}
   { \bf  p}_{F}\Vert_F^2$ at the optimal ${\bf P}_F^{\star}= \lambda_P{ \bf p}_F{\bf p}_F^H$.  Based on this, it can  be  found from \eqref{WPno0}  that  $R_S$  is a monotonically   non-decreasing function of $\lambda_P$. Additionally, when  
the   cross-interference  constraint ${\text{CR3}}$ is inactive for each  problem,    it is clear that $\lambda_P$ is only subject to the  constraint ${\text{CR2}}$ for problems \eqref{SRM}, \eqref{WSRM}, \eqref{OUTRM},  or   constraints ${\text{AR2}}$  and  ${\text{AR3}}$ for problem \eqref{ARSRM2}.    Based on  \eqref{WPno0} and  problem \eqref{ARSRM2}, we   readily infer that both  ${\text{CR2}}$  and   ${\text{AR3}}$ are  active  at the optimal ${ \bf P}_F^{\star}$ for maximizing  secrecy  rate $R_S$ and  artificial noise power, respectively.  Inspired by  this  conclusion,    we  further    formulate  the common   subproblem   for optimizing  PB energy covariance matrices  ${ \bf   W}_P$  in  different cases of FBS-EVE CSI  with the inactive  ${\text{CR3}}$     as
\begin{align}\label{WPno}
  &\underset{\tau, { \bf   W}_P\succeq{ \bf   0}}
{\text{max}}~~~\tau \xi\text{tr}({ \bf    H}_F{ \bf    W}_P{ \bf    H}_F^H),~~~{\rm{s.t.}}~~\text{CR1:}~0\le\tau\le 1,~\text{tr}({ \bf   W}_P) \le P_P.
\end{align}  

According to  \cite[Lemma H.1.h]{Inequal},
we  readily derive   the closed-form solution  ${ \bf   W}_P^{\star}$ to    problem \eqref{WPno} as   ${ \bf   W}_P^{\star}=P_P{{\bf w }_P{\bf w }_P^{ H }}$ with ${\bf w }_P=\nu_{\max}({\bf H}_F^H{\bf H}_F)$.


\end{appendices}

\end{document}